\documentclass[letterpaper, 10pt,draftcls]{IEEEtran}  %

\onecolumn 
\IEEEoverridecommandlockouts                              

\usepackage{graphics} 
\usepackage{epsfig} 
\usepackage{mathptmx} 
\usepackage{times} 
\usepackage{amsmath} 
\usepackage{amssymb}  
\usepackage{amsthm}
\usepackage{subcaption}
\usepackage{comment}
\usepackage[export]{adjustbox}
\usepackage{tikz}
\usepackage{pgfplots}
\usepackage[outdir=./]{epstopdf}
\usepackage[outline]{contour}

\newtheorem{thm}{Theorem}
\newtheorem{prop}{Proposition}
\newtheorem{cor}[prop]{Corollary}
\newtheorem{lem}[prop]{Lemma}

\newtheorem{example}{Example}
\newcommand{\expphi}{\left(e^{2\pi i \phi} \right)}
\newcommand{\mmse}{\mathrm{mmse}}

\tikzstyle{int}=[draw, fill=blue!10, minimum height = 1cm, minimum width=1.5cm,thick ]
\tikzstyle{small_int}=[draw, fill=blue!10, minimum height = 0.5cm, minimum width=0.7cm,thick ]

\title{\LARGE \bf
The Distortion Rate Function of \\ Cyclostationary Gaussian Processes
}

\author{Alon Kipnis$^{*}$, Andrea J. Goldsmith$^{*}$
 and Yonina C. Eldar$^{\dagger}$ 
\thanks{$^{*}$Department of Electrical Engineering, Stanford University, CA}
\thanks{$^{\dagger}$Department of EE, Technion, Israel Institute of Technology, Haifa, Israel}
\thanks{This paper was presented in part at the 2014 international symposium on information theory, Honolulu, HI, USA.}
}

\begin{document}
\graphicspath{{../Figures/}}

\maketitle

\thispagestyle{plain}
\pagestyle{plain}

\begin{abstract}
A general expression for the distortion rate function (DRF) of cyclostationary Gaussian processes in terms of their spectral properties is derived. This expression can be seen as the result of orthogonalization over the different components in the polyphase decomposition of the process. We use this expression to derive, in a closed form, the DRF of several cyclostationary processes arising in practice. We first consider the DRF of a combined sampling and source coding problem. It is known that the optimal coding strategy for this problem involves source coding applied to a signal with the same structure as one resulting from pulse amplitude modulation (PAM). Since a PAM-modulated signal is cyclostationary, our DRF expression can be used to solve for the minimal distortion in the combined sampling and source coding problem. We also analyze in more detail the DRF of a source with the same structure as a PAM-modulated signal, and show that it is obtained by reverse waterfilling over an expression that depends on the energy of the pulse and the baseband process modulated to obtain the PAM signal. This result is then used to study the information content of a PAM-modulated signal as a function of its symbol time relative to the bandwidth of the underlying baseband process. In addition, we also study the DRF of sources with an amplitude-modulation structure, and show that the DRF of a narrow-band Gaussian stationary process modulated by either a deterministic or a random phase sine-wave equals the DRF of the baseband process. 
\end{abstract}


\section{INTRODUCTION}
\label{sec:Intro}
The distortion rate function (DRF) describes the average minimal distortion achievable in sending an information source over a rate-limited noiseless link. Sources with memory posses an inherent statistical dependency that can be exploited in the context of data compression. However, not many closed-form expressions for the DRF of such sources are known, and those are usually limited to the class of stationary processes. Two notable exceptions are the DRFs of the Wiener process, derived by Berger \cite{berger1970information}, and of auto-regressive Gaussian processes, derived by Gray \cite{1054470}. Indeed, information sources are rarely stationary in practice, and source coding techniques that are based on stationary assumptions about the source will likely achieve poor performance if the source has time-varying statistics. \par
\emph{Cyclostationary processes} (CS) (also known as \emph{periodically correlated} processes or \emph{block-stationary} processes) are a class of non-stationary processes whose statistics are invariant to time shifts by integer multiples of a given time constant, denoted as the \emph{period} of the process. As described in the survey by Gardner \cite{Gardner:2006}, CS processes have been used in many fields to model periodic time-variant phenomena. In particular, they arise naturally in synchronous communication where block coding and modulation by periodic signals are used. Spectral properties of CS processes, which will be used in our derivations, are also reviewed in \cite{Gardner:2006} and in the references therein. \\

In this paper we analyze the DRF of CS Gaussian processes. Our study of the DRF of CS processes can be motivated by a very simple example: we are interested to find the DRF of the process obtained by modulating a continuous-time Gaussian stationary process $U(\cdot)$ by a cosine wave with random phase $\Phi$, namely
\begin{equation} \label{eq:intro_random_phase}
X_{\Phi}(t) = \sqrt{2}U(t)\cos\left(2\pi f_0t + \Phi \right),\quad t\in \mathbb R,
\end{equation}
where $\Phi$ is uniformly distributed over $[0,2\pi)$. This process is commonly given as an example of a wide-sense stationary process in signal processing textbooks (e.g. \cite[Ex. 8.18]{2009introduction}). Note that due to the random phase, $X_\Phi(\cdot)$ is not Gaussian and in fact is non-ergodic. It seems that in the context of rate-distortion theory, the spectrum of $X_\Phi(\cdot)$ can only be used to derive an upper bound on its DRF given by the DRF of a Gaussian stationary process with the same second order statistics \cite[Thm. 4.6.5]{berger1971rate}. The theory of asymptotic mean stationary (AMS) processes \cite{gray2009probability} implies that the DRF of $X_{\Phi}(t)$ is given by the DRF of each one of its ergodic components \cite[Thm. 11.3.1]{gray1990entropy}, corresponding to different values of the phase $\varphi \in [0,2\pi)$. Each ergodic component satisfies a source coding theorem which allows us to evaluate its DRF using an optimization over probability distributions subject to a mutual information rate constraint. One might think this decomposition might provide us with a recipe to evaluate the DRF of $X_\Phi(\cdot)$ by averaging over the DRF of the process $X_\varphi(\cdot)$, obtained by fixing the phase in $X_{\Phi}(\cdot)$. However, while the process $X_\varphi(\cdot)$ is Gaussian, it is no longer stationary -- but rather CS. While an expression for the DRF of CS processes is known \cite{Berger1968254,gray1990entropy}, it seems that the only existing mechanism for evaluating this DRF is by the Karhunen-Lo\`{e}ve (KL) expansion \cite{gallager1968information}. In this method it is required to solve for the eigenvalues of a Fredholm integral equation for each finite blocklength, and use a waterfilling expression over these eigenvalues. The DRF is finally obtained in the limit as the size of the blocklength goes to infinity. This evaluation, however, does not exploit the special block periodicity of the CS source. Moreover, it does not provide any intuition on the optimal source coding technique in terms of spectral properties of the source. In contrast, the DRF for a stationary Gaussian process is obtained by waterfilling over its power spectral density, which provides clear intuition about how the source code represents each frequency component of the signal \cite{Berger1998}. \par
In this work we derive an expression for the DRF of Gaussian CS processes which uses their spectral properties, and therefore generalizes the waterfilling expression for the DRF of Gaussian stationary processes derived by Pinsker \cite{1056823}. This expression is obtained by considering the polyphase components of the process, which can be seen as a set of stationary processes that comprise the CS process \cite{vetterli2014foundations}. We show that the DRF of a discrete-time CS can be obtained in a closed form by orthogonalizing over these components at each frequency band. For continuous-time CS processes, we obtain an expression which is based on increasingly fine discrete-time approximations of the continuous-time signal. The DRF evaluated for these approximations converges to the DRF of the continuous-time process under mild conditions on its covariance function. \par

The main results of this paper are divided into two parts. In the first part we derive a general expression for evaluating the DRF of a second order Gaussian CS process in terms of its spectral properties. This expression is given in the form of a reverse waterfilling solution over the eigenvalues of a spectral density matrix defined in terms of the \emph{time-varying spectral density} of the source. For discrete-time Gaussian processes, the size of this matrix equals the discrete period of the source. We extend our result to Gaussian CS processes in continuous-time by taking increasingly finer discrete time approximations. The resulting expression is a function of the eigenvalues of an infinite matrix. In addition, we derive a lower bound on the DRF which can be obtained without evaluation of the matrix eigenvalues. We show that this bound is tight when the polyphase components of the process are highly correlated. \par
In the second part of the paper we use our general DRF expression to study the distortion-rate performance in more specific cases. Specifically:
\begin{itemize}
\item We derive a closed form expression for the DRF of a process with a pulse-amplitude modulation (PAM) signal structure. We show how this expression can be used to derive the minimal distortion in estimating a stationary Gaussian process from a rate-limited version of its sub-Nyquist samples.
\item We study the effect of the symbol time in PAM on the information content of the modulated signal at the output of the modulator. 
\item 
We evaluate in a closed form the DRF of a Gaussian stationary narrowband process modulated by a deterministic cosine wave. We show that the DRF of the modulated process equals that of the baseband stationary Gaussian process provided the latter is narrowband. We further conclude that the stationary, non-Gaussian and non-egodic process given by \eqref{eq:intro_random_phase} above has DRF identical to the DRF of the modulated process without the random phase. These two results imply that the DRF of the stationary non-Gaussian amplitude modulated process is strictly smaller than the DRF of a Gaussian stationary process with the same second order statistics.
\end{itemize}


The rest of this paper is organized as follows: in Section~\ref{sec:background} we review concepts and notation from the theory of CS processes and rate distortion theory. Our main results are given in Section~\ref{sec:main_result}, where we derive an expression for the DRF of a Gaussian CS process. In Section \ref{sec:bound} we derive a lower bound on this DRF. In Section~\ref{sec:applications} we explore applications of out main result in various special cases. Concluding remarks are provided in Section~\ref{sec:conclusions}.

\section{Definitions and Notations}
\label{sec:background}
\subsection{Cyclostationary Processes}
Throughout the paper, we consider zero mean Gaussian processes in both discrete and continuous time. We use round brackets to denote a continuous time index and square brackets for a discrete time index, i.e. 
\[
X(\cdot)=\left\{X(t),\,t\in \mathbb R \right\},
\]
and
\[
X[\cdot]=\left\{X[n],\,n\in \mathbb Z \right\}.
\]
Matrices and vectors are denoted by bold letters. \\

The statistics of a zero mean Gaussian process $X(\cdot)$ is specified in terms of its autocorrelation function\footnote{ In \cite{gardner1994cyclostationarity} and in other references, the symmetric auto-correlation function 
\[
\tilde{R}_X(t,\tau) \triangleq \mathbb E\left[ X(t+\tau/2) X(t-\tau/2)\right] = R_X(t-\tau/2,\tau),
\]
the corresponding CPSD $\hat{\tilde{S}}^n_X(f)$ and TPSD $\tilde{S}^{~t}_X(f)$, are used. The conversion between $\hat{S}^n(f)$ and the symmetric CPSD
is given by $\hat{\tilde{S}}^n_X(f) = \hat{S}_X^n(f-n/(2T_0))$.
}
\[
R_X(t,\tau)\triangleq \mathbb E\left[X(t+\tau)X(t) \right].
\]
If in addition the autocorrelation function is periodic in $t$ with a fundamental period $T_0$,
\[
R_X(t+T_0,\tau)= R_X(t,\tau),
\]
then we say that $X(\cdot)$ is a \emph{cyclostationary process} or simply \emph{cyclostationry}\footnote{It is customary to distinguish between wide-sense cyclostationarity which relates only to the second order statistics of the process, and strict-sense cyclostationarity which relates to the finite order statistics of the process \cite[Ch. 10.4]{papoulis2002probability}. Both definitions coincide in the Gaussian case.} \cite{Gardner:2006, gardner1994cyclostationarity}. We also assume that $R_X(t,\tau)$ is bounded and Riemann integrable on $[0,T_0] \times \mathbb R$, and therefore
\[
\sigma_X^2 = \lim_{T\rightarrow \infty}  \frac{1}{2T} \int_{-T}^T \mathbb E X(t)^2 dt = \frac{1}{T_0} \int_0^{T_0} R_X(t,0) dt
\]
is finite. 
\par
Suppose that $R_X(t,\tau)$ has a convergent Fourier series representation in $t$ for almost any $\tau \in \mathbb R$. Then the statistics of $X(\cdot)$ is uniquely determined by the \emph{cyclic autocorrelation} (CA) function:
\begin{equation} \label{eq:CA_def}
\hat{R}_{X}^n(\tau) \triangleq \frac{1}{T_0} \int_{-T_0/2}^{T_0/2} R_X(t,\tau)e^{-2\pi i n t/T_0}dt,\quad n\in \mathbb Z.
\end{equation}
The Fourier transform of $\hat{R}_X^n(\tau)$ with respect to $\tau$ is denoted as the \emph{cyclic power spectral density} (CPSD) function:
\begin{equation} \label{eq:CPSD_def}
\hat{S}_X^n (f) =\int_{-\infty}^\infty \hat{R}_X^n(\tau) e^{-2\pi i \tau f}d\tau,\quad -\infty\leq f \leq \infty.
\end{equation}

If $\hat{S}_X^n(f)$ is identically zero for all $n \neq 0$, then $R_X(t,\tau)=R_X(0,\tau)$ for all $0\leq t \leq T$ and the process $X(\cdot)$ is stationary. In such a case $S_X(f)\triangleq \hat{S}_X^0(f)$ is the \emph{power spectral density} (PSD) function of $X(\cdot)$. The \emph{time-varying power spectral density} (TPSD) function \cite[Sec. 3.3]{gardner1994cyclostationarity} of $X(\cdot)$ is defined by the Fourier transform of $R_X(t,\tau)$ with respect to $\tau$, i.e.
\begin{equation} \label{eq:TPSD_def}
S^{~t}_X(f)\triangleq \int_{-\infty}^\infty R_X(t,\tau)e^{-2\pi i f\tau }d\tau.
\end{equation}
The Fourier series representation implies that
\begin{equation} \label{eq:CPSD_series_relation}
S_X^{~t}(f)=\sum_{n\in \mathbb Z} \hat{S}_X^n(f)e^{2\pi i n t/T_0}.
\end{equation}

Associated with every cyclostationary process $X(\cdot)$ with period $T_0$ is a set of stationary discrete time processes $X^t[\cdot]$, $0\leq t \leq T_0$, defined by 
\begin{equation} \label{eq:PC_def}
X^t[n]=X\left(T_0n+t\right),\quad n\in \mathbb Z.
\end{equation}
These processes are called the \emph{polyphase components} (PC) of the cyclostationary process $X(\cdot)$.
The cross-correlation function of $X^{t_1}[\cdot]$ and $X^{t_2}[\cdot]$ is given by  
\begin{align}
R_{X^{t_1}X^{t_2}}[n,k] & =\mathbb E \left[ X[T_0(n+k)+t_1]X[T_0n+t_2] \right] \nonumber \\
& = R_X\left(T_0n+t_2,T_0k+t_1-t_2 \right) \nonumber \\
& = R_X\left(t_2,T_0k+t_1-t_2\right).
\end{align}
Since $R_{X^{t_1}X^{t_2}}[n,k]$ depends only on $k$, this implies that $X^{t_1}[\cdot]$ and $X^{t_2}[\cdot]$ are jointly stationary. The PSD of $X^t[\cdot]$ is given by 
\begin{align}
S_{X^t}\left(e^{2\pi i \phi} \right) \triangleq &\sum_{k\in \mathbb Z}R_{X^t X^t}[0,k] e^{-2\pi i \phi k} \nonumber \\
= & \sum_{k\in \mathbb Z}R_X\left(t,T_0k \right) e^{-2\pi i \phi k},\quad -\frac{1}{2}\leq \phi \leq \frac{1}{2}.  \label{eq:stationary_PSD}
\end{align}

Using the spectral properties of sampled processes, we can use \eqref{eq:stationary_PSD} and \eqref{eq:CPSD_series_relation} to connect the functions $S_{X^t}\left(e^{2\pi i \phi} \right)$ and the CPSD of $X(\cdot)$ as follows:
\begin{align}
S_{X^t}\left(e^{2\pi i \phi} \right) 
= & \frac{1}{T_0}\sum_{k\in \mathbb Z} S_X^{~t}\left(\frac{\phi-k}{T_0} \right) \nonumber \\
= & \frac{1}{T_0}\sum_{k\in \mathbb Z} \sum_{n\in \mathbb Z} \hat{S}_X^n\left(\frac{\phi-k}{T_0}  \right)e^{2\pi i n t/T_0}. \nonumber 
\end{align}
More generally, for $t_1,t_2 \in [0,T_0]$ we have
\begin{align}
S_{X^{t_1}X^{t_2}}\left(e^{2\pi i \phi} \right) 
= &\sum_{k\in\mathbb Z} R_{X^{t_1}X^{t_2}}[0,k] e^{-2\pi i k \phi} \label{eq:CPSD_PSD_cont}  \\
= & \frac{1}{T_0}\sum_{k\in \mathbb Z} S_X^{~t_2}\left(\frac{\phi-k}{T_0}  \right) e^{2\pi i(t_1-t_2)\frac{\phi-k}{T_0}} \nonumber \\
= & \frac{1}{T_0}\sum_{k\in \mathbb Z} \sum_{m\in \mathbb Z} \hat{S}_X^m\left(\frac{\phi-k}{T_0}  \right)e^{2\pi i\left( m \frac{t_2}{T_0}+ \frac{t_1-t_2}{T_0}(\phi-k)\right)} \nonumber. 
\end{align}

We now turn to briefly describe the discrete-time counterpart of the CA, CPSD, TPSD and the polyphase components defined in \eqref{eq:CA_def}, \eqref{eq:CPSD_def}, \eqref{eq:TPSD_def} and \eqref{eq:PC_def}, respectively. \par
A discrete time zero mean Gaussian process ${X}[\cdot]$ is said to be CS with period $M\in \mathbb N$ if its covariance function 
\[
R_{{X}}[n,k]=\mathbb E\left[ {X}[n+k] {X}[n] \right]
\]
is periodic in $k$ with period $M$. For $m=0,\ldots,M$, the $m^\textrm{th}$ cyclic autocorrelation (CA) function of ${X}[\cdot]$ is defined as
\[
\hat{R}_{{X}}^m[k] \triangleq \sum_{n=0}^{M-1}R_{{X}}[n,k]e^{-2\pi i nm/M}.
\]
The $m^\textrm{th}$ CPSD function is then given by
\begin{align*}
\hat{S}_{{X}}^m\left(e^{2\pi i\phi} \right)\triangleq & \sum_{k\in \mathbb Z} \hat{R}_{{X}}^m[k] e^{-2\pi i \phi k},
\end{align*}
and the discrete TPSD function is 
\[
S_{{X}}^n \left(e^{2\pi i \phi} \right)\triangleq  \sum_{k\in \mathbb Z} R_{{X}}[n,k]e^{-2\pi i \phi k}.
\]
Finally, we have the discrete time Fourier transform relation
\[
S_{{X}}^n\left(e^{2\pi i \phi} \right)= \frac{1}{M} \sum_{m=0}^{M-1} \hat{S}_{{X}}^m\left(e^{2\pi i\phi} \right) e^{2\pi i \phi n m/M}.
\]
\par
The $m$-th stationary component $\bar{X}^m[\cdot]$, $0\leq m \leq M-1$ of ${X}[\cdot]$ is defined by 
\begin{equation} \label{eq:stationary_component_discrete}
{X}^m[n]\triangleq {X}[Mn+m], \quad n\in \mathbb Z.
\end{equation}
For $0\leq m,r,n\leq M-1$ and $k\in \mathbb Z$ we have
\begin{align}
R_{{X}^{m}{X}^r}[n,k] & = \mathbb E \left[{X}^m[n+k] {X}^r[n]\right] \nonumber \\
& = \mathbb E \left[{X}[Mn+Mk+m] {X}[Mn+r]\right] \nonumber \\
& = R_{{X}}[Mn+r,Mk+m-r] \nonumber \\
& = R_{{X}}[r,Mk+m-r]. \label{eq:discerete_R_cross}
\end{align}
Using properties of multi-rate signal processing:
\begin{align}
 S_{{X}^{m}{X}^r}\left(e^{2\pi i \phi} \right)
= & \sum_{k\in \mathbb Z} R_{{X}}[r,Mk+m-r] e^{-2\pi i k \phi} 
\nonumber \\ 
= & \frac{1}{M}\sum_{n=0}^{M-1}S^r_{{X}}\left(e^{2\pi i \frac{\phi-n}{M}} \right)e^{2\pi i (m-r) \frac{\phi-n}{M}}. \label{eq:PSD_discrete}
\end{align}
The discrete-time counterpart of \eqref{eq:CPSD_PSD_cont} is then
\begin{equation}
\label{eq:PC_cov_alternative}
S_{{X}^m {X}^r} \left(e^{2\pi i \phi}\right)= \frac{1}{M}\sum_{k=0}^{M-1} \sum_{n=0}^{M-1} \hat{S}_{{X}}^n\left(e^{2\pi i \frac{\phi-k}{M}} \right)e^{2\pi i\frac{nr+(m-r)(\phi-k)}{M} }.
\end{equation}
The functions $S_{{X}^m {X}^r}\left(e^{2\pi i \phi} \right)$, $0\leq m,r \leq M-1$ define an $M\times M$ matrix $\mathbf S_{{X}} \expphi$ with $(m+1,r+1)^\textrm{th}$ entry $S_{{X}^m {X}^r}\left(e^{2\pi i \phi} \right)$. This matrix completely determines the statistics of ${X}[\cdot]$, and can be seen as the PSD matrix associated with the stationary vector valued process ${\mathbf X}^M[n]$ defined by the stationary components of ${X}[\cdot]$:
\begin{equation} \label{eq:associated_stationary}
{\mathbf X}^M[n]\triangleq \left({X}^0[n],\ldots,{X}^{M-1}[n] \right),\quad n\in \mathbb Z.
\end{equation}
We denote the autocorrelation matrix of ${\mathbf X}^M[\cdot]$ as the PSD-PC matrix. Note that the $(r+1,m+1)^\mathrm{th}$ entry of the PSD-PC matrix is given by \eqref{eq:discerete_R_cross}. \\


\subsection{Examples}
We present two important modulation models which result in CS processes.

\begin{example}[amplitude modulation (AM)] \label{ex:AM}
Given a Gaussian stationary process $U(\cdot)$ with PSD $S_U(f)$, consider the process
\[
X_{AM}(t) = \sqrt{2} U(t) \cos \left(2\pi f_0 t +\varphi \right),
\]
where $f_0 >0$ and $\varphi \in [0,2\pi)$ are deterministic constants. This process is CS with period $T_0 = f_0^{-1}$ and CPSD \cite[Eq. 41]{1096820}
\[
\hat{S}_{AM}^m(f)= \frac{1}{2} \begin{cases} S_U(f+f_0)+ S_U(f-f_0), & m=0, \\
S_U(f\mp f_0)e^{\pm 2i \varphi }, & m\pm 2, \\
0, & \textrm{otherwise}.
\end{cases}
\]
This leads to the TPSD
\begin{align}
\label{eq:AM_SCD}
S_{X}^{~t}(f)=\frac{1}{2} S_U(f+f_0)&(1+e^{-2 \left(2\pi i f_0 t+\varphi \right)} ) \\ & +
 \frac{1}{2} S_U(f-f_0)(1+ e^{2 \left(2\pi i f_0 t+\varphi \right)}) \nonumber.
\end{align}

\end{example}

\begin{example}[pulse-amplitude modulation (PAM)] \label{ex:PAM}
Consider a Gaussian stationary process $U(\cdot)$ modulated by a deterministic signal $p(t)$ as follows:
\begin{equation} \label{eq:pam_def}
X_{PAM}(t) = \sum_{n\in \mathbb N} U(nT_0)p(t-nT_0).
\end{equation}
This process is CS with period $T_0$ and CPSD \cite[Eq. 49]{1096820}
\begin{equation} \label{eq:PAM_CPSD}
\hat{S}_{PAM}^n(f) =  \frac{1}{T_0} P\left(f \right) P^*\left(f-\frac{n}{T_0} \right)S_U\left(f\right), \quad n\in \mathbb Z,
\end{equation}
where $P(f)$ is the Fourier transform of $p(t)$ and $P^*(f)$ is its complex conjugate. If $T_0$ is small enough such that the support of $P(f)$ is contained within the interval $\left(-\frac{1}{2T_0}, \frac{1}{2T_0} \right)$, then $\hat{S}_{PAM}^n(f)=0$ for all $n\neq 0$, which implies that $X_{PAM(\cdot)}$ is stationary. 
\end{example}

\subsection{The Distortion-Rate Function}
For a fixed $T>0$, let $X_T$ be the reduction of $X(\cdot)$ to the interval $[-T,T]$. Define the distortion between two waveforms $x(\cdot)$ and $y(\cdot)$ over the interval $[-T,T]$ by 
\begin{equation} \label{eq:dist_cont}
d_T \left( x(\cdot), y(\cdot) \right) \triangleq \frac{1}{2T} \int_{-T}^T  \left(x(t) - y(t)\right)^2dt.
\end{equation}
We expand $X_T$ by a Karhunen-Lo\`{e}ve (KL) expansion \cite[Ch 9.7]{gallager1968information} as
\begin{equation} \label{eq:KL_expansion}
X_T(t) = \sum_{k=1}^\infty X_k f_k(t), \quad -T\leq t \leq T,
\end{equation}
where $\left\{f_k\right\}$ is a set of orthogonal functions over $[-T,T]$ satisfying the Fredholm integral equation
\begin{equation}
\lambda_k f_k(t) = \frac{1}{2T} \int_{-T}^T K_X(t,s) f_k(s) ds,\quad t\in[-T,T],
\label{eq:fredholm}
\end{equation}
with corresponding eigenvalues $\left\{ \lambda_k\right\}$, and where 
\[
K_X(t,s) \triangleq \mathbb E X(t) X(s) = R_X(s,t-s). 
\]
Assuming a similar expansion as \eqref{eq:KL_expansion} to an arbitrary random waveform $Y_T$, we have
\[
\mathbb E d_T(X_T,Y_T) = \frac{1}{2T} \int_{-T}^T \mathbb E \left( X(t) - Y(t) \right)^2 dt = \sum_{n=-\infty}^\infty \mathbb E \left(X_n - Y_n \right)^2. 
\]
The mutual information between $X(\cdot)$ and $Y(\cdot)$ on the interval $[-T,T]$ is defined by
\[
 I_T \left( X(\cdot), Y(\cdot) \right) \triangleq \frac{1}{2T} \lim_{N\rightarrow \infty} I\left( \mathbf X_{-N}^N ; \mathbf Y_{-N}^N  \right),
\]
where $\mathbf X_{-N}^N = \left( X_{-N},\ldots, X_{N} \right)$, $\mathbf Y_{-N}^N = \left( Y_{-N},\ldots, Y_{N} \right)$ and the $X_n$s and $Y_n$s are the coefficients in the KL expansion of $X(\cdot)$ and $Y(\cdot)$, respectively. \par 
Denote by $\mathcal P_T$ the set of joint probability distributions $P_{X,\hat{X}}$ over the waveforms $\left(X(\cdot), \hat{X}(\cdot) \right)$, such that the marginal of $X(\cdot)$ agrees with the original distribution, and the average distortion $\mathbb E d_T \left( X(\cdot), \hat{X}(\cdot) \right)$ does not exceed $D$. The rate-distortion function (RDF) of $X(\cdot)$ is defined by
\[
R(D) =  \lim_{T\rightarrow \infty} R_T(D),
\]
where 
\[
R_T(D) =  \inf I_T\left( X(\cdot); \hat{X}(\cdot) \right), 
\]
and the infimum is over the set $\mathcal P_T$. It is well known that $R(D)$ and $R_T(D)$ are non-decreasing convex functions of $D$ \cite{berger1971rate}, and therefore continuous in $D$ over any open interval. We define their inverse function as the distortion-rate functions $D(R)$ and $D_T(R)$, respectively. We note that by its definition, $D(R)$ is bounded from above by the average power of $X(\cdot)$ over a single period:
\begin{align*}
\sigma_X^2 & \triangleq \lim_{T\rightarrow \infty} \frac{1}{2T} \int_{-T}^T \mathbb E X^2(t) dt  = \lim_{T\rightarrow \infty} \frac{1}{2T} \int_{-T}^T R_X(t,0) dt \\
& = \frac{1}{T_0} \int_0^{T_0} R_X(t,0) dt = \hat{R}^0_X(0).
\end{align*}

For Gaussian processes, we have the following parametric representation for $R_T(D)$ or $D_T(R)$ \cite[Eq. 9.7.41]{gallager1968information}
\begin{subequations}
\label{eq:DRF_KL}
\begin{align}
D_T(\theta) & = \sum_{k=1}^\infty \min \left\{\theta,\lambda_k \right\} \label{eq:DRF_KL_D} \\
R_T(\theta) & = \frac{1}{2}\sum_{k=1}^\infty \log^+\left( \lambda_k/\theta \right) \label{eq:DRF_KL_R},
\end{align}
\end{subequations}
where $\log^+ x \triangleq \max\left\{ \log x,0\right\}$. 
 \par
In the discrete-time case the DRF is defined in a similar way as in the continuous-time setting described above by replacing the continuous-time index in \eqref{eq:dist_cont}, \eqref{eq:KL_expansion} and \eqref{eq:fredholm}, and by changing integration to summation. Since the KL transform preserves norm and mutual information, this definition of the DRF in the discrete-time case is consistent with standard developments for the DRF of a discrete-time source with memory as in \cite[Ch. 4.5.2]{berger1971rate}. Note that with these definitions, the continuous-time distortion is measured in MSE per time unit while the discrete-time distortion is measured in MSE per source symbol. Similarly, in continuous-time, $R$ represents bitrate, i.e., the number of bits per time unit. In the discrete-time setting we use the notation $\bar{R}$ to denote bits per source symbol. \par

Since the distribution of a zero-mean Gaussian CS process with period $T_0$ is determined by its second moment $R_X(t,\tau)$, we observe that such processes are $T_0$-ergodic and therefore \emph{block-ergodic} as defined in \cite[Def. 1]{Berger1968254}. It follows that a source coding theorem that associates $D(R)$ with the optimal MSE performance attainable in encoding $X(\cdot)$ at rate $R$ is obtained from the main result of \cite{Berger1968254}. Specifically for the discrete-time case, it is shown in \cite[Exc. 6.3.1]{gray2009probability} that CS processes belong to the class of asymptotic mean stationary process (AMS) \cite{gray2009probability}, where a source coding theorem for AMS processes can be found in \cite{gray2011entropy}. \\

\subsection{Problem Formulation: Evaluating the DRF}
In the special case in which $X(\cdot)$ is stationary, it is possible to obtain $D(R)$ without explicitly solving the Fredholm equation \eqref{eq:fredholm} or evaluating the KL eigenvalues: in this case, the density of these eigenvalues converges to the PSD $S_X\left(f\right)$ of $X(\cdot)$. This leads to the celebrated \emph{reverse waterfilling} expression for the DRF of a stationary Gaussian process, originally derived by Pinsker \cite{1056823}: 
\begin{subequations}
\label{eq:SKP_cont}
\begin{equation} \label{eq:SKP_R_cont}
R(\theta) = \frac{1}{2} \int_{-\infty}^\infty \log^+\left[S_X\left(f \right)/\theta \right] df.
\end{equation}
\begin{equation} \label{eq:SKP_D_cont}
D(\theta) = \int_{-\infty}^\infty \min \left\{ S_X\left(f \right),\theta \right\}d\phi.
\end{equation}
\end{subequations}

The discrete-time version of \eqref{eq:SKP_cont} is given by 
\begin{subequations}
\label{eq:SKP}
\begin{equation} \label{eq:SKP_R}
\bar{R}(\theta) = \frac{1}{2} \int_{-\frac{1}{2}}^\frac{1}{2} \log^+\left[S_X\left(e^{2\pi i \phi} \right)/\theta \right] d\phi.
\end{equation}
\begin{equation} \label{eq:SKP_D}
D(\theta) = \int_{-\frac{1}{2}}^\frac{1}{2} \min \left\{ S_X\left(e^{2\pi i \phi} \right),\theta \right\}d\phi.
\end{equation}
\end{subequations}
Equations \eqref{eq:SKP_cont} and \eqref{eq:SKP} define the distortion as a function of the rate through a joint dependency on the water level parameter $\theta$. \par
We note that stationarity is not a necessary condition for the existence of a density function for the eigenvalues in the KL expansion. For example, such a density function is known for the Wiener process \cite{berger1970information} which is a non-stationary process. \\

The main problem we consider in this paper is the evaluation of $D(R)$ for a general CS Gaussian process. In principle, this evaluation can be obtained by computing the KL eigenvalues in \eqref{eq:fredholm} for each $T$, using \eqref{eq:DRF_KL} to obtain $D_T(R)$ and finally taking the limit as $T$ goes to infinity. For general CS processes, however, an easy way to describe the density of the KL eigenvalues is in general unknown. As a result, the evaluation of the DRF directly by the KL eigenvalues usually does not lead to a closed-form solution. In the next section we derive an alternative representation for the function $D(R)$ which is based on an approximation of the kernel $K_X(t,s)$ used in \eqref{eq:fredholm}.

\section{Main Results \label{sec:main_result}}
In this section we derive our main results with respect to  an expression for the DRF of a Gaussian CS which does not involve the solution of the Fredholm integral equation \eqref{eq:fredholm}. \\

Our first observation is that in the discrete-time case, the DRF of a Gaussian CS process can be obtained by an expression for the DRF of a vector Gaussian stationary source. This expression is an extension of \eqref{eq:SKP}, which was derived in \cite[Eq. (20) and (21)]{1628751} and is given as follows:
\begin{subequations}
\label{eq:vector_RD}
\begin{align}
D_{ {\mathbf X}}\left(\theta\right)&=\frac{1}{M}\sum_{m=1}^{M}\int_{-\frac{1}{2}}^{\frac{1}{2}}\min\left\{ \lambda_{m}\left(e^{2\pi i \phi} \right),\theta\right\} d\phi \label{eq:vector_RD_D} \\
R\left(\theta\right)&= \frac{1}{M} \sum_{m=1}^{M}\int_{-\frac{1}{2}}^{\frac{1}{2}}\frac{1}{2}\log^{+}\left[\lambda_m\left(e^{2\pi i \phi} \right)/\theta\right]d\phi,\label{eq:vector_RD_R}
\end{align}
\end{subequations}
where $\lambda_1\left(e^{2\pi i \phi} \right),..., \lambda_{M}\left(e^{2\pi i \phi} \right)$
are the eigenvalues of the PSD matrix $\mathbf{S}_{{\mathbf X}}\left(e^{2\pi i \phi} \right)$ at frequency $\phi$. We have the following result:
\begin{thm} \label{thm:main_result_discrete}
Let ${X}[\cdot]$ be a discrete-time Gaussian cyclostationary process with period $M\in \mathbb N$. The distortion rate function of ${X}[\cdot]$ is given by
\begin{subequations}
\label{eq:main_result_discrete}
\begin{align} 
D(\theta)& =\frac{1}{M}\sum_{m=1}^{M}\int_{-\frac{1}{2}}^{\frac{1}{2}} \min\left\{\lambda_m\left(e^{2\pi i \phi }\right),\theta \right\}d\phi  \label{eq:main_result_discrete_D} \\
\bar{R}(\theta)& =\frac{1}{2M} \sum_{m=1}^{M} \int_{-\frac{1}{2}}^{\frac{1}{2}}\log^+ \left[\lambda_m \left(e^{2\pi i \phi }\right)/\theta \right]d\phi, \label{eq:main_result_discrete_R}
\end{align}
\end{subequations}
where $\lambda_1\expphi \leq \ldots \leq \lambda_{M}\expphi $ are the eigenvalues of the PSD-PC matrix with $\left(m+1,r+1\right)^\textrm{th}$ entry given by 
\begin{equation}
\mathbf S_{{X}^m {X}^r} \expphi = \frac{1}{M}\sum_{n=0}^{M-1}S^r_{{X}}\left(e^{2\pi i \frac{\phi-n}{M}} \right)e^{2\pi i (m-r) \frac{\phi-n}{M}}. \label{eq:thm_main_entries}
\end{equation}
\end{thm}

\begin{proof}
A full proof can be found in Appendix \ref{sec:main_result_discrete}. The idea is to use the polyphase decomposition \eqref{eq:PSD_discrete} and the stationary vector valued process ${\mathbf X}^M[\cdot]$ defined in \eqref{eq:associated_stationary}. The PSD matrix of the process is shown to coincide with the PSD-PC matrix of $X[\cdot]$. The proof shows that the DRF of $X[\cdot]$ coincides with the DRF of $\mathbf{X}^M[\cdot]$. The result then follows by applying \eqref{eq:vector_RD} to ${\mathbf X}^M[\cdot]$. 
\end{proof}

Equation~\eqref{eq:main_result_discrete} has the waterfilling interpretation illustrated in Fig.~\ref{fig:waterfilling_eigenvalues}: the DRF is obtained by setting a single water-level over all eigenvalues of \eqref{eq:thm_main_entries}. These eigenvalues can be seen as the PSD of $M$ independent processes obtained by the orthogonalization of the PC of $X[\cdot]$. Compared to the limit in the discrete-time version of the KL expansion \eqref{eq:DRF_KL}, expression \eqref{eq:main_result_discrete} exploits the CS structure of the process by using its spectral properties. These spectral properties capture information on the entire time-horizon and not only over a finite blocklength as in the KL expansion. \par
The following theorem explains how to extend the above evaluation to the continuous-time case.

\begin{figure}
\begin{center}
\includegraphics[scale=0.5, clip = true, trim = 2.5cm 0cm 1cm 0cm]{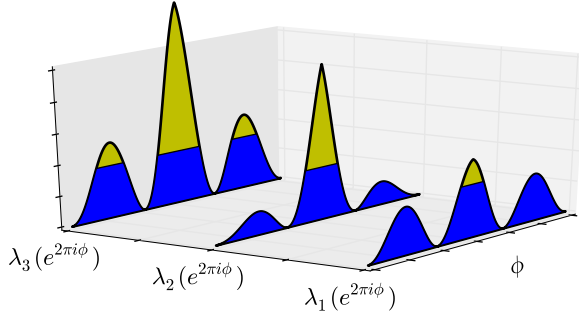}
\caption{\label{fig:waterfilling_eigenvalues} Waterfilling interpretation of \eqref{eq:main_result_discrete} for $M=3$. The blue and the yellow parts are associated with equations \eqref{eq:main_result_discrete_D} and \eqref{eq:main_result_discrete_R}, respectively.}
\end{center}
\end{figure}

\begin{thm} \label{thm:main_result}
Let $X(\cdot)$ be a Gaussian cyclostationary process with period $T_0$ and correlation function $R_X(t,\tau)$ Lipschitz continuous in its second argument. For a given $M\in \mathbb N$, denote
\begin{subequations}
\label{eq:main_result_cont}
\begin{align}
D_M(\theta_M) & =\frac{1}{M}\sum_{m=1}^{M}\int_{-\frac{1}{2}}^{\frac{1}{2}} \min\left\{\lambda_m\left(e^{2\pi i \phi }\right),\theta_M \right\}d\phi  \label{eq:main_D} \\
R(\theta_M) & =\frac{1}{2T_0} \sum_{m=1}^{M} \int_{-\frac{1}{2}}^{\frac{1}{2}}\log^+ \left[\lambda_m \left(e^{2\pi i \phi }\right)/\theta_M \right]d\phi, \label{eq:main_R}
\end{align}
\end{subequations}
where $\lambda_1\expphi \leq \ldots \leq \lambda_M \expphi $ are the eigenvalues of the matrix $\mathbf S_X \expphi $ with its $(m+1,r+1)^\textrm{th}$ entry given by 
\begin{align}
\frac{1}{T_0}\sum_{k\in \mathbb Z } & S_{X}^{rT_0/M} \left(\frac{\phi-k}{T_0} \right) e^{2\pi i (m-r) \frac{\phi-k}{M}}
\label{eq:main_result_entries} \\
& = \frac{1}{T_0}\sum_{k\in \mathbb Z } \sum_{n\in \mathbb Z}  \hat{S}^n_X \left(\frac{\phi-k}{T_0} \right) e^{2\pi i \frac{ n r + (m-r) (\phi-k) } {M}}. \nonumber
\end{align}
Then the limit of $D_M$ in $M$ exists and the distortion-rate function of $X(\cdot)$ is given by 
\begin{equation}
\label{eq:main_thm_DR}
D(R)=\lim_{M\rightarrow \infty} D_M\left(\theta_M(R)\right).
\end{equation}
\end{thm}

\subsubsection*{Proof sketch} The proof idea is to use a CS discrete-time process that approximates $X(\cdot)$. This approximation becomes tighter as $M$ increases, so that the limit in \eqref{eq:main_thm_DR} converges to the DRF of the continuous-time process coincides. The proof details are given in Appendix~\ref{sec:proof_main}. \\

\subsection*{Discussion}
The expression \eqref{eq:main_result_cont} is obtained by taking the limit in \eqref{eq:main_result_discrete} over the time-period of a discrete-time CS process, where the code rate $R$ is appropriately adjusted to bits per time unit. Although \eqref{eq:main_result_cont} only provides the DRF in terms of a limit, this limit is associated with the intra-cycle time resolution and not with the time horizon as in \eqref{eq:DRF_KL}. This fact allows us to express the DRF in terms of the spectral properties of the process, which captures `memory' in the process over the entire time horizon. \par
We note that limits of the form \eqref{eq:main_thm_DR} have been obtained in closed-form using Szeg\H o's Toeplitz distribution theorem \cite[Section 5.2]{grenander1958toeplitz} when the underlying process is stationary and the matrix considered is Toeplitz \cite{gallager1970information, berger1971rate,1054470} or block Toeplitz \cite{BlockToeplitz,ShannonMeetsNyquist}. Unfortunately, the matrix in \eqref{eq:main_result_entries} is not Toeplitz or block Toeplitz so Szeg\H o's theorem is not applicable. 
In the following section we provide a few examples where the limit in \eqref{eq:main_thm_DR} can be obtained in closed form which lead to a closed form expression for the DRF. \\

Expression \eqref{eq:main_result_cont} can be seen as the extension to CS of the waterfilling expression \eqref{eq:SKP_cont} derived for stationary processes. While the latter can be understood as the limiting result of coding over orthogonal frequency bands \cite{Berger1998}, expression \eqref{eq:main_result_cont} implies that the DRF for CS processes is the result of two orthogonalization procedures: (1) over the PC inside a cycle, which is associated with the eigenvalues decomposition of the PSD-PC, and (2) over different frequency bands of the stationary processes resulting from the first orthogonalization. \par
The decomposition of the process into its stationary PCs can be further exploited to derive a lower bound on the DRF, which become tight as these PC become highly correlated. This bound is explored in the next section.

\section{Lower Bound \label{sec:bound}}
In this section we derive a lower bound on the DRF of a Gaussian CS process. This lower bound is expressed only in terms of the PCs of the process and does not require the eigenvalue evaluation of Theorems~\ref{thm:main_result_discrete} and \ref{thm:main_result}. The basis for this bound is the following proposition, which holds for any source distribution and distortion measure (although we will consider here only the quadratic Gaussian case).
\begin{prop}\label{prop:lower_bound} Let $\mathbf{X}[\cdot]$ be a vector valued process of dimension $M$. The distortion-rate function
of $\mathbf{X}[\cdot]$ satisfies
\begin{equation}
D_{\mathbf{X}}\left(R\right)\geq\frac{1}{M}\sum_{m=0}^{M-1}D_{X_m}\left(R\right).\label{eq:lower_bound}
\end{equation}
\end{prop}
\begin{proof}
Any rate $R$ code for the process $\mathbf X[\cdot]$ induces a rate $R$ code on each of the coordinates $X_m[\cdot]$, $m=0,...,M-1$. At each coordinate, this code cannot achieve lower distortion than the optimal rate $R$ code for that coordinate.
\end{proof}



Proposition~\ref{prop:lower_bound} applied to Gaussian CS processes leads to the following result:
\begin{prop}\label{prop:lower_bound_disc}
Let ${X}[\cdot]$ be a discrete-time Gaussian CS process with period $M\in \mathbb N$. The distortion rate function of ${X}[\cdot]$ satisfies
\begin{equation} \label{eq:lower_bound_discrete}
D (\bar{R})\geq \frac{1}{M}\sum_{m=0}^{M-1}\int_{-\frac{1}{2}}^{\frac{1}{2}} \min \left \{ S_{X^m} \left(e^{2\pi i \phi} \right) ,\theta_m \right\} d\phi,
\end{equation}
where for each $m=0,\ldots,M-1$, $\theta_m$ satisfies
\begin{equation}
\label{eq:lower_bound_R}
\bar{R}(\theta_m)=\frac{1}{2} \int_{-\frac{1}{2}}^{\frac{1}{2}} \log^+ \left[S_{X^m}\left(e^{2\pi i \phi} \right)/\theta_m \right] d\phi.
\end{equation}
\end{prop}
Here
\[
S_{X^m}\left(e^{2\pi i \phi} \right) \triangleq S_{X^mX^m}\left(e^{2\pi i \phi} \right)  =  \frac{1}{M}\sum_{n=0}^{M-1}S^m_{{X}}\left(e^{2\pi i \frac{\phi-n}{M}} \right)
\]
is the PSD of the $m^\textrm{th}$ PC of $X[\cdot]$.\\

\subsubsection*{Proof} The claim is a direct application of Proposition~\ref{prop:lower_bound} to our case of a discrete-time CS process: the summands on the RHS of \eqref{eq:lower_bound_discrete} are the individual DRF of the PCs $X^m[\cdot]$, $m=0,\ldots,M-1$, of $X[\cdot]$ obtained by \eqref{eq:SKP}. \\

Proposition~\ref{prop:lower_bound_disc} can be extended to the continuous-time case by approximating the outer integral in \eqref{eq:lower_bound_cont} by finite sums. This yields the following result:
\begin{prop}\label{prop:lower_bound_cont}
Let $X(\cdot)$ be a continuous-time Gaussian cyclostationary process with period $T_0>0$ and correlation function $R_X(t,\tau)$ Lipschitz continuous in its second argument. The distortion rate function of $X(\cdot)$ satisfies
\begin{align} \label{eq:lower_bound_cont} 
D(R) \geq \frac{1}{T_0} \int_0^{T_0} \int_{-\frac{1}{2}}^{\frac{1}{2}} \min \left \{ \sum_{n\in \mathbb Z} S_X^{~t}\left(\frac{\phi-n}{T_0} \right),\theta_t \right\} d\phi dt,
\end{align}
where for each $0\leq t \leq T_0$, $\theta_t$ satisfies
\begin{equation} \label{eq:lower_bound_cont_R}
R(\theta_t)=\frac{1}{2T_0} \int_{-\frac{1}{2}}^{\frac{1}{2}} \log^+ \left[ \sum_{n\in \mathbb Z} S_X^{~t} \left(\frac{\phi-n}{T_0} \right) /\theta_t \right] d\phi.
\end{equation}
\end{prop}

\subsubsection*{Proof}
See Appendix~\ref{sec:proof_bound_cont}. 


The bound \eqref{eq:lower_bound_discrete} is obtained by averaging the minimal distortion at rate $R$ in describing each one of the PCs of $X(\cdot)$. For each such component $X^t[\cdot]$ there is an associated water level $\theta_t$ obtained by solving \eqref{eq:lower_bound_cont_R} for $\theta_t$. For $R=0$, $\theta_t$ is always bigger than the essential supremum of 
\[
S_{X^t}\left( e^{2\pi i \phi}\right) = \frac{1}{T_0} \sum_{n\in \mathbb Z} S_X^{~t} \left(\frac{\phi-n}{T_0} \right),
\]
so the RHS of \eqref{eq:lower_bound_discrete} equals the average over the total power of each one of the PCs of ${X}(\cdot)$ which are summed to $\sigma_X^2 = D_{{X}}(0)$. On the other hand, if $R \rightarrow \infty$ then $\theta_t\rightarrow 0$ for all $t\in[0,T_0]$, and again equality holds in \eqref{eq:lower_bound_discrete}. That is, the bound is tight in the two extremes of $R=0$ and $R\rightarrow \infty$. \par

From a source coding point of view, the bound \eqref{eq:lower_bound_cont} can be understood as if a source code of rate $R$ is applied to each of the PCs of $X(\cdot)$ individually. On the other hand, the DRF in \eqref{eq:main_result_cont}  is obtained by applying a single rate $R$ code to describe all these PC simultaneously. As a result, the bound is tight only when all the PCs are maximally correlated, i.e. when a single PC determines the rest of them. A case where the latter hold is shown in the following example. 
\begin{example}[equality in \eqref{eq:lower_bound_cont}]
Let $X(\cdot)$ be the PAM signal of Example~\ref{ex:PAM} where the pulse $p(t)$ is given by
\[
p(t) = \begin{cases} 1 & 0\leq t < T_0, \\
0 & \textrm{otherwise}.
\end{cases}
\]
The sample path of $X(\cdot)$ has a staircase shape illustrated in Figure~\ref{fig:staricase}. This process is equivalent to the discrete time process $\bar{U}[\cdot] \triangleq \left\{U(nT_0),\,n\in \mathbb Z\right\}$ both in information rate and squared norm per period $T_0$, which is enough to conclude that $D_{X}(R) = D_{U}(R T_0)$. Indeed, the PCs in this case are maximally correlated, in the sense that a realization of $X^0[\cdot]=\left\{X(nT_0),\,n\in \mathbb Z\right\}$ determines the value of $X^\Delta[\cdot] = \left\{X((n+\Delta)T_0),\,n\in \mathbb Z \right\}$ for all $0\leq \Delta<1$. In addition, for all $0\leq t \leq T_0$ we have
\[
S_{X}^{~t} \expphi = S_{X}^0 \expphi =  \sum_{n\in \mathbb Z} S_U\left(\frac{\phi-n}{T_0} \right),
\]
where the latter is the PSD of the discrete time process $\bar{U}[\cdot]$, so \eqref{eq:SKP_cont} implies that the RHS of \eqref{eq:lower_bound_cont} is the DRF of $\bar{U}[\cdot]$. We therefore conclude that the DRF of $X(\cdot)$ is given by the RHS of \eqref{eq:lower_bound_cont}.
\end{example}

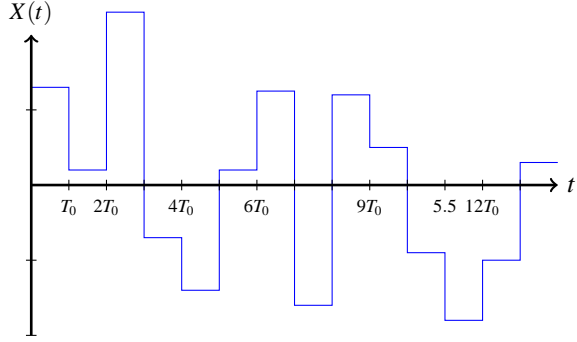
\begin{figure}
\begin{center}
\begin{tikzpicture}
\draw[color = blue] (0,1.3) -- (0.5,1.3) -- (0.5,0.2) -- (1,0.2) -- (1,2.3) -- (1.5,2.3) -- (1.5,-0.7) -- (2,-0.7) -- (2,-1.4) -- (2.5,-1.4) -- (2.5,0.2) -- (3,0.2)--(3,1.25)--(3.5,1.25)--(3.5,-1.6)--(4,-1.6) -- (4,1.2) -- (4.5,1.2) -- (4.5,0.5) -- (5,0.5)--(5,-0.9) -- (5.5,-0.9) -- (5.5,-1.8) --(6,-1.8) -- (6,-1) --(6.5,-1)--(6.5,0.3) -- (7,0.3);
  
\foreach \x/\xtext in {0/,0.5/{T_0},1/{2T_0},1.5/,2/{4T_0},2.5/,3/{6T_0},3.5/,4/,4.5/{9T_0},5/,5.5,6/{12T_0},6.5/}
 \draw[shift={(\x,0)}] (0pt,2pt) -- (0pt,-2pt) node[below] {\scriptsize $\xtext$};;
  \foreach \y/\ytext in {-2/,-1/,1/}
    \draw[shift={(0,\y)}] (2pt,0pt) -- (-2pt,0pt) node[left] {$\ytext$};

\draw[->,line width=1pt]  (0,0)--(7,0) node[right] {$t$};
\draw[->,line width=1pt]  (0,-2)--(0,2) node[above] {\small $X(t)$};
\end{tikzpicture}
\caption{\label{fig:staricase}  
An example of a continuous-time PAM process that attains equality in \eqref{eq:lower_bound_cont}. 
}
\end{center}
\end{figure}

\section{Applications\label{sec:applications}}
In this section we apply the expression obtained in Theorem~\ref{thm:main_result} to study the distortion-rate performance of a few CS processes that arise in practice.

\tikzstyle{int}=[draw, fill=blue!10, minimum height = 1cm, minimum width=1.5cm,thick ]

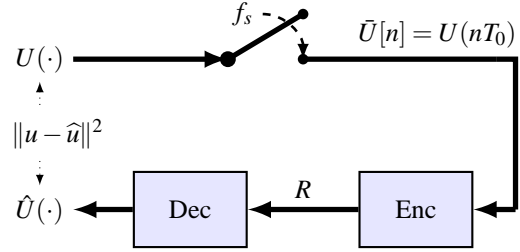
\begin{figure}
\begin{center}
\begin{tikzpicture}[node distance=2cm,auto,>=latex]
  \node at (0,0) (source) {$U(\cdot)$} ;
  \node [coordinate, right of = source,node distance = 2.5cm] (smp_in) {};
  \node [coordinate, right of = smp_in,node distance = 1cm] (smp_out){};
	\node [coordinate,above of = smp_out,node distance = 0.6cm] (tip) {};
\fill  (smp_out) circle [radius=2pt];
\fill  (smp_in) circle [radius=3pt];
\fill  (tip) circle [radius=2pt];
\node[left,left of = tip, node distance = 0.8 cm] (ltop) {$f_s$};
\draw[->,dashed,line width = 1pt] (ltop) to [out=0,in=90] (smp_out.north);
\node [right of = smp_out, node distance = 2.7cm] (right_edge) {};
\node [below of = right_edge] (right_b_edge) {};

\node [right] (dest) [below of=source]{$\hat{U}(\cdot)$};
\node [int] (dec) [right of=dest, node distance = 2cm] {$\mathrm{Dec}$};
\node [int] (enc) [right of = dec, node distance = 3cm]{$\mathrm{Enc}$};

  \draw[-,line width=2pt] (smp_out) -- node[above, xshift = 0.5cm] {$\bar{U}[n]=U(nT_0)$} (right_edge);
  \draw[-,line width = 2]  (right_edge.west) -| (right_b_edge.east);
    \draw[->,line width = 2]  (right_b_edge.east) -- (enc.east);
   \draw[->,line width=2pt] (enc) -- node[above] {$R$} (dec);

\draw[<->,dotted] (source) -- node[xshift = -0.5cm, fill = white] {$\|u-\widehat{u}\|^2$} (dest);
   \draw[->,line width=2pt] (dec) -- (dest);
    \draw[->, line width = 2pt] (source) -- (smp_in);
    \draw[line width=2pt]  (smp_in) -- (tip);
    
\end{tikzpicture}
\caption{\label{fig:sub_nyquist_model} Combined sampling and source coding system model.}
\end{center}
\end{figure}

\subsection{Combined Sampling and Source Coding \label{subsec:combined}}
We begin with the distortion-rate performance in the combined sampling and source coding problem considered in \cite{Kipnis2014}. This problem is described by the system of Figure~\ref{fig:sub_nyquist_model}: the source $U(\cdot)$ is a Gaussian stationary process with a known PSD $S_U(f)$. The source is uniformly sampled at rate $f_s = T_s^{-1}$, resulting in the discrete time process $\bar{U}[\cdot]$ defined by $\bar{U}[n]=U(n/f_s)$. The process $\bar{U}[\cdot]$ is then encoded at rate $R$ bits per time unit. The goal is to estimate the source $U(\cdot)$ from its sampled and encoded version under a quadratic distortion. We denote by the function $D_{U|\bar{U}}(f_s,R)$ the minimal distortion attainable in this estimation, where the minimization is over all collections of encoders and decoders operating at bitrate $R$. Note that if $U(\cdot)$ is sampled above its Nyquist rate, then there is no loss of information in the sampling operation, and we get
\[
D_{U|\bar{U}}(f_s,R) = D_{U}(R),
\]
where $D_U(R)$ is found by \eqref{eq:SKP_cont}. Therefore, the case of most interest is that of sub-Nyquist sampling of $U(\cdot)$. In what follows we use Theorem~\ref{thm:main_result} to derive $D_{U|\bar{U}}(f_s,R)$ in closed form.\\

Our first observation is that the combined sampling and source coding problem of Fig.~\ref{fig:sub_nyquist_model} can be seen as an indirect source coding problem \cite{1057738}: the distortion is measured with respect to the process $U(\cdot)$, but a different process, namely $\bar{U}[\cdot]$, is available to the encoder. Wolf and Ziv \cite{1054469} have shown that the optimal source coding scheme under quadratic distortion for this class of problems is obtained as follows: the encoder first obtains the minimal mean square error (MMSE) estimate of the unseen source, and then an optimal source code is applied to describe this estimated sequence to the decoder. In the setting of Figure~\ref{fig:sub_nyquist_model}, this implies that $D_{U|\bar{U}}(f_s,R)$ is attained by first obtaining the MMSE estimate 
\[
\widetilde{U}(t) = \mathbb E\left[U(t)|\bar{U}[\cdot] \right]
\]
at the encoder, and then solving a standard source coding problem with the process $\widetilde{U}(\cdot)$ as the process to which the source code is applied. Moreover, this scheme implies that the distortion decomposes into 2 parts:
\begin{align} \label{eq:decomp}
D_{U|\bar{U}}(f_s,R) = {\mmse}(U|\bar{U}) + D_{\widetilde{U}}(R),
\end{align}
where ${\mmse}({U|\bar{U}})$ is the MMSE in estimating $U(\cdot)$ from $\bar{U}[\cdot]$, and $D_{\widetilde{U}}(R)$ is the DRF of the process $\widetilde{U}(\cdot)$. \par

Standard linear estimation techniques \cite{815501}  leads to 
\[
\widetilde{U}(t) = \sum_{n\in\mathbb Z} \bar{U}[n]w(t-nT_0) = \sum_{n\in\mathbb Z} U(nT_0)w(t-n T_0),
\]
where the Fourier transform of $w(t)$ given by
\begin{equation}
\label{eq:hybrid_filter}
W(f)=\frac{ S_U(f)}{\sum_{k\in \mathbb Z} S_U(f-k/T_0)}.
\end{equation}
Moreover, the error in this estimation is 
\begin{equation} \label{eq:mmse_expression}
{\mmse}(U|\bar{U}) = \int_{-\infty}^\infty S_U(f) df- \int_{-\frac{1}{2T_0}}^\frac{1}{2T_0} \widetilde{S}_{W}(f) df,
\end{equation}
where
\begin{align} 
\widetilde{S}_W(f) & = \sum_{k\in \mathbb Z} \left| W(f-k/T_0) \right|^2 S_U(f-k/T_0) \label{eq:sampling_RD_J}. 
\end{align}

\tikzstyle{sint}=[draw, fill=blue!10, minimum height = 0.5cm, minimum width=0.8cm,thick ]
\tikzstyle{sum}=[circle, fill=blue!10, draw=black,line width=1pt,minimum size = 0.5cm, thick ]
\tikzstyle{ssum}=[circle, fill=blue!10,draw=black,line width=1pt,minimum size = 0.1cm]
\tikzstyle{int1}=[draw, fill=blue!10, minimum height = 0.5cm, minimum width=1cm,thick ]
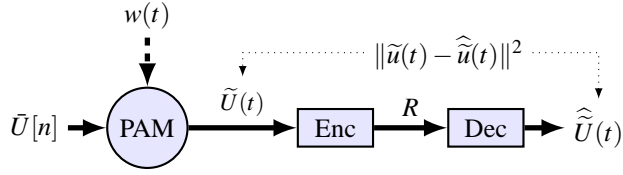
\begin{figure}
\begin{center}
\begin{tikzpicture}[node distance=2cm,auto,>=latex]
  \node at (0,0) (source) {$\bar{U}[n]$} ;
  \node [sum, right of = source,node distance = 1.5cm] (PAM) {PAM};
  \node [above of = PAM, node distance = 1.5cm] (pulse) {$w(t)$}; 

\node [int1] (enc) [right of = PAM, node distance = 2.5cm] {$\mathrm{Enc}$};
\node [int1] (dec) [right of=enc, node distance = 2cm] {$\mathrm{Dec}$};

\node [right] (dest) [right of = dec, node distance = 1.5cm]{$\widehat{\widetilde{U}}(t)$};

\draw[->,line width=2pt] (enc) -- node[above] {$R$} (dec);

\node [above of = dec, node distance = 1cm, xshift = -0.5cm] (dist) {$\|\widetilde{u}(t)-\widehat{\widetilde{u}}(t)\|^2$};

\draw[->,line width=2pt, dashed] (pulse) -- (PAM) ;
\draw[->, line width = 2pt] (source) -- (PAM);
\draw[->, line width = 2pt]  (PAM) -- node[above] (tilde) {\small $\widetilde{U}(t)$} (enc);
\draw[->,line width=2pt] (dec) -- (dest);
\draw[->, dotted] (dist) -| (dest);
\draw[->, dotted] (dist) -| (tilde);
\end{tikzpicture}
\caption{\label{fig:sub_nyquist_solution} Minimal distortion in Figure~\ref{fig:sub_nyquist_model} is obtained by PAM followed by an optimal source code for the output of this modulation.}
\end{center}
\end{figure}
We conclude from the above that $D_{\widetilde{U}}(R)$, and therefore $D_{U|\bar{U}}(f_s,R)$, is obtained by solving a source coding problem for an information source with a PAM signal structure, illustrated in Figure~\ref{fig:sub_nyquist_solution}. Since Example~\ref{ex:PAM} implies that such a signal is CS with period $T_s = f_s^{-1}$, we can apply Theorem~\ref{thm:main_result} in order to evaluate this DRF. By doing so, we obtain the following general result:
\begin{prop}[DRF of PAM-modulated signals] \label{prop:pam}
Let $X_{PAM}(\cdot)$ be defined by
\begin{equation} \label{eq:pam_def2}
 X_{PAM}(t) = \sum_{n\in \mathbb Z} U(nT_0)p(t-nT_0), \quad t\in \mathbb R,
\end{equation}
where $U(\cdot)$ is a Gaussian stationary process with\footnote{Although we only use the value of $U(t)$ at $t \in \mathbb Z T_0$, it is convenient to treat $U(\cdot)$ as continuous-time source so that the expressions emerging have only continuous-time spectrum.} PSD $S_U(f)$ and $p(t)$ is an analog deterministic signal with $\int_{-\infty}^\infty \left|p(t)\right|^2dt<\infty$ and Fourier transform $P(f)$. Assume moreover, that the covariance function $R_{X_{PAM}}(t,\tau)$ of $X_{PAM}(\cdot)$ is Lipschitz continuous in its second argument. The distortion-rate function of $X_{PAM}(\cdot)$ is given by
\begin{subequations} \label{eq:DRF_PAM}
\begin{align}
D(\theta) & =  \frac{1}{T_0}\int_{-\frac{1}{2T_0}}^{\frac{1}{2T_0}} \min\left\{\widetilde{S}(f),\theta \right\} df \label{eq:DRF_PAM_D} \\
R(\theta) & = \frac{1}{2} \int_{-\frac{1}{2T_0}}^{\frac{1}{2T_0}} \log^+\left[ \widetilde{S}(f) /\theta \right] df,
\end{align}
\end{subequations}
where
\begin{equation} \label{eq:J_def}
\widetilde{S}(f) \triangleq  \sum_{k\in \mathbb Z} \left|P(f-k/T_0)\right|^2 S_U(f-k/T_0).
\end{equation}
\end{prop}
\begin{proof}
See Appendix~\ref{app:pam_proof}. 
\end{proof}

Proposition~\ref{prop:pam} applied to the process $\widetilde{U}(\cdot)$ implies that its DRF $D_{\widetilde{U}}(R)$ is given by waterfilling over the function 
\[
J(f) \triangleq \sum_{k\in \mathbb Z} \left|W(f-k f_s)\right|^2 S_U(f-k f_s). 
\]
As a result, we obtain from \eqref{eq:decomp} and \eqref{eq:DRF_PAM} the following expression for the minimal distortion in the combined sampling and source coding problem:
\begin{subequations}
\label{eq:sampling_RD}
\begin{align}
D_{U|\bar{U}}(f_s,R) = \mmse(U|\bar{U}) +  \frac{1}{T_0} \int_{-\frac{1}{2T_0}}^\frac{1}{2T_0} \min\left\{ J(f) ,\theta \right\}df, 
\end{align}
where 
\begin{equation}
R(\theta) =\frac{1}{2}  \int_{-\frac{1}{2T_0}}^\frac{1}{2T_0} \log^+ \left[ J(f) /\theta \right]df.
\end{equation}
\end{subequations} 



\subsection{Information Content of Signals with PAM Structure}
Proposition~\ref{prop:pam} provides a general closed-form expression for the DRF of Gaussian processes with a PAM structure. In this subsection we use this expression to study the effect of the PAM of \eqref{eq:pam_def2} on the distortion-rate curve of the signal  $X_{PAM}(\cdot)$ at the output of the modulator. 
Assuming that two processes have the same energy over time, the process with lower DRF can be described by fewer bits per second to the same distortion level. It is therefore intuitive to think about the DRF as a measure of the \emph{information content} of the process\footnote{This notion is made precise by the notion of $\epsilon$-entropy \cite{posner1967epsilon,kolmogorov1956certain}}. \par

If we assume that the source for the symbols in the PAM is a Gaussian stationary process $U(\cdot)$, the output of the PAM of \eqref{eq:pam_def2} can be seen as a non-ideal reconstruction of $U(\cdot)$ from its uniform samples using pulses of shape $p(t)$, as illustrated in Figure~\ref{fig:PAM_example}. Since the randomness in $X_{PAM}(\cdot)$ is only due to $U(\cdot)$, we expect $X_{PAM}(\cdot)$ to have a smaller information content than $U(\cdot)$. In addition, we expect the information content of $X_{\Phi}(\cdot)$ to increase with the sampling rate $1/T_0$, and reach a saturation as this sampling rate exceeds the Nyquist rate of $U(\cdot)$. 
Indeed, when $1/T_0$ is higher than the Nyquist rate of $U(\cdot)$, the support of $S_U(f)$ is contained within $\left(-\frac{1}{2T_0}, \frac{1}{2T_0} \right)$. In this case, expression \eqref{eq:DRF_PAM} implies that the DRF of $X_{PAM}(\cdot)$ is obtained by waterfilling over the function 
\begin{equation} \label{eq:change_of_coordinate}
\widetilde{S}(f) = \left| P(f) \right|^2 S_U(f).
\end{equation}
That is, the effect of the modulation in super-Nyquist sampling is identical to the effect of a linear filter with frequency response $P(f)$ applied to $U(\cdot)$. This filtering can be understood as a linear transformation of the coordinates \cite[Ch. 22]{Shannon1948} represented by the frequency components. Assuming that $P(f)$ does not change the support of \eqref{eq:change_of_coordinate} (that is, the change in `coordinates' is invertible), the process $U(\cdot)$ can be recovered from $X_{PAM}(\cdot)$ with zero mean-square error. When the sampling frequency $1/T_0$ goes below the Nyquist rate of $U(\cdot)$, perfect recovery of $U(\cdot)$ is in general not possible. Intuitively, in this case we expect $X_{PAM}(\cdot)$ to contains less information than $U(\cdot)$, and hence we should be able to describe it under the same normalized distortion level as $U(\cdot)$ using fewer bits per time unit. 
A quantitative evaluation of this effect of the sampling rate is given in Figure~\ref{fig:PAM_example}, where the DRF of $X_{PAM}(\cdot)$ is compared to the DRF of $U(\cdot)$ for three sub-Nyquist sampling rates. Examples for the realization of $X_{PAM}(\cdot)$ and $U(\cdot)$ using sub- and super- Nyquist sampling rates are given in Figure~\ref{fig:PAM_realizations}.
\begin{figure}
\begin{center}
\begin{tikzpicture}
\node[int] at (0,0) {\includegraphics[scale=0.2]{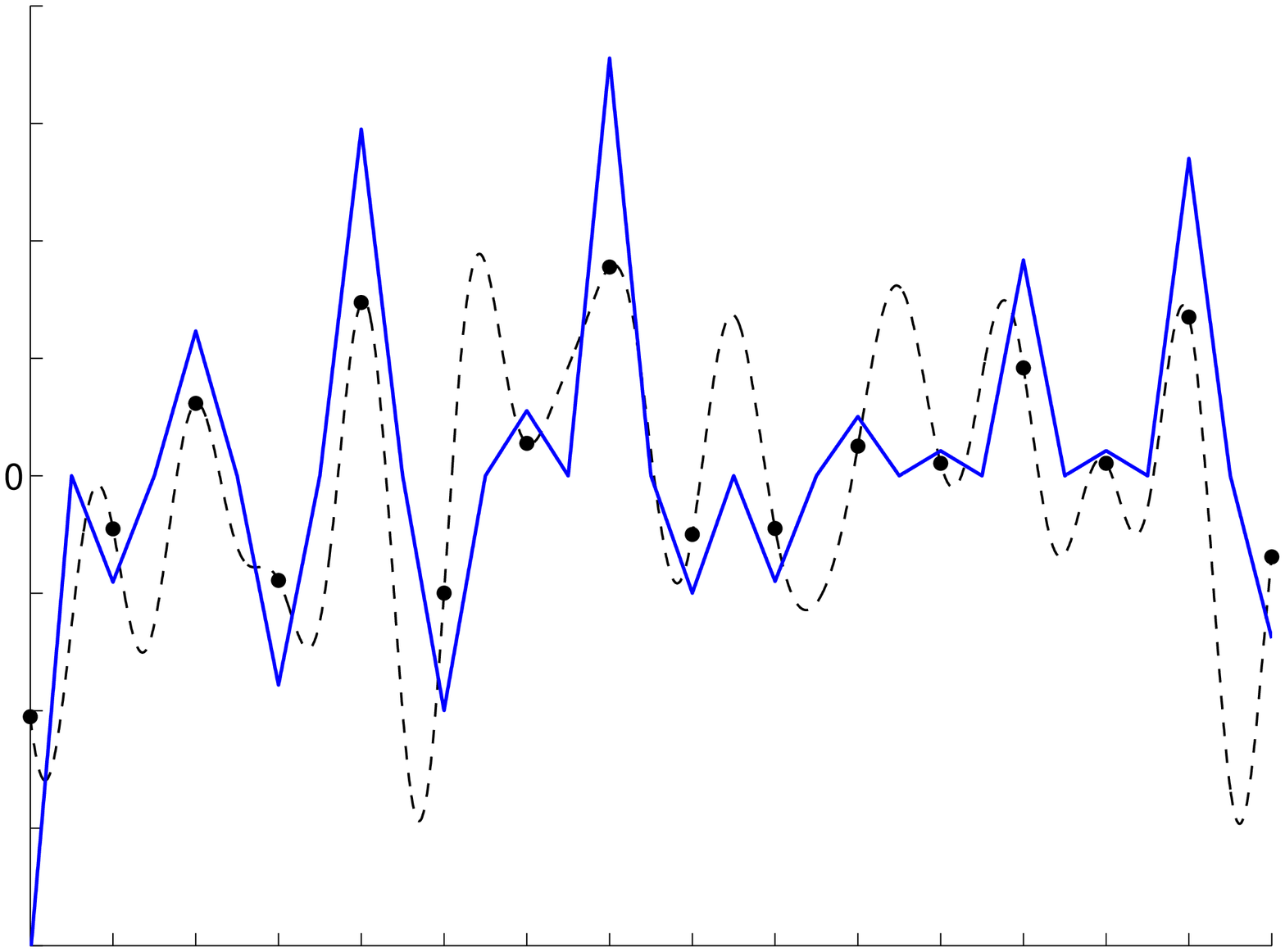}};
\node[int] at (4.5,0) {\includegraphics[scale=0.2]{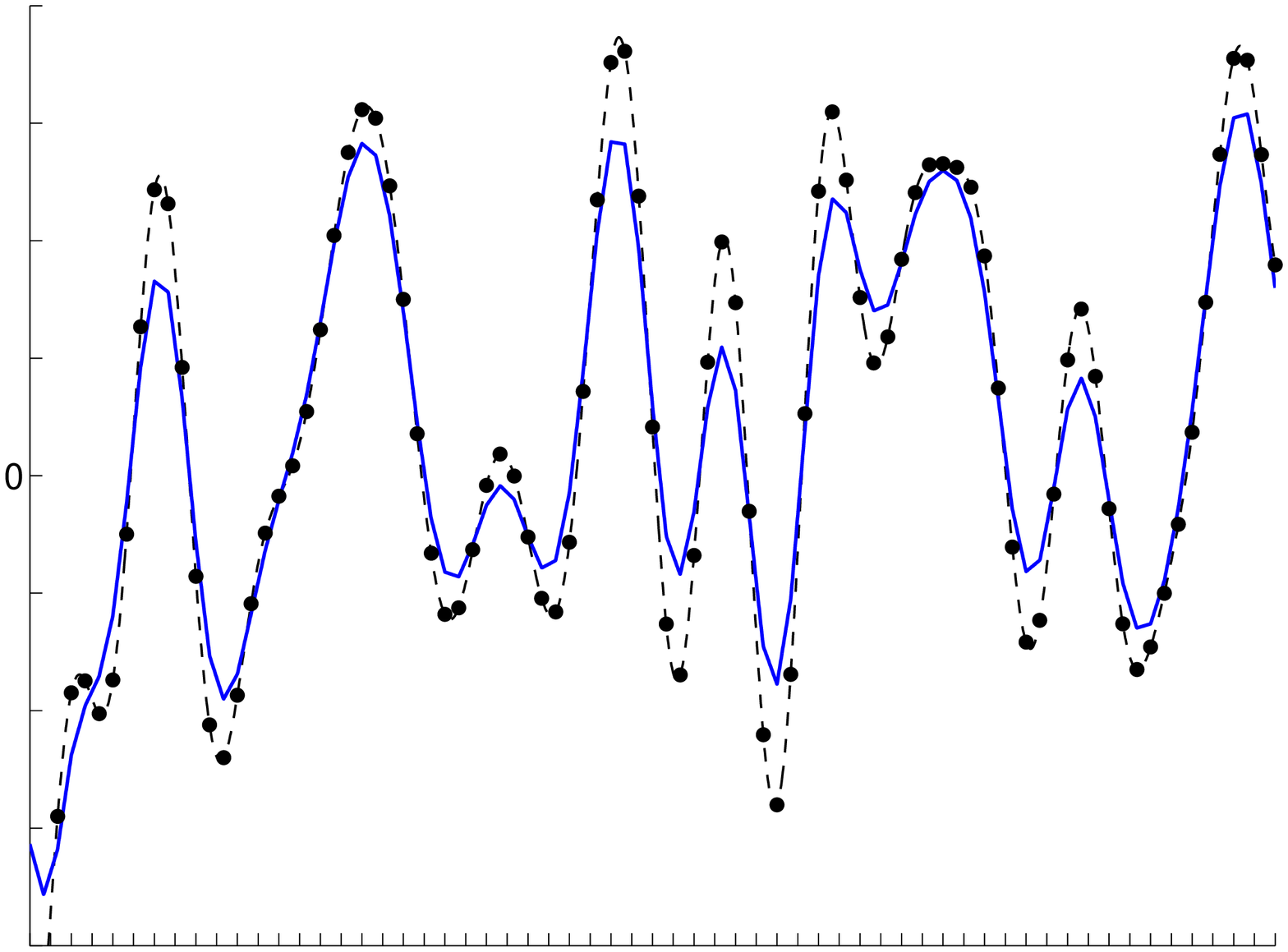}};
\end{tikzpicture}
\caption{Two realizations over time of the PAM process $X_{PAM}(\cdot)$ (blue) and the baseband process $U(\cdot)$ (dashed) with the PSD and pulse shape given in Figure~\ref{fig:PAM_example}, corresponding to sub-Nyquist (left) and super-Nyquist (right) sampling rates. Figure~\ref{fig:PAM_example} below shows that for the same target distortion, the PAM-modulated process on the left is easier to describe than the PAM-modulated process on the right.
\label{fig:PAM_realizations}}
\end{center}
\end{figure}
\begin{figure}
\begin{center}
\begin{tikzpicture}
\node at (0.215,-0.25) {\includegraphics[scale=0.4]{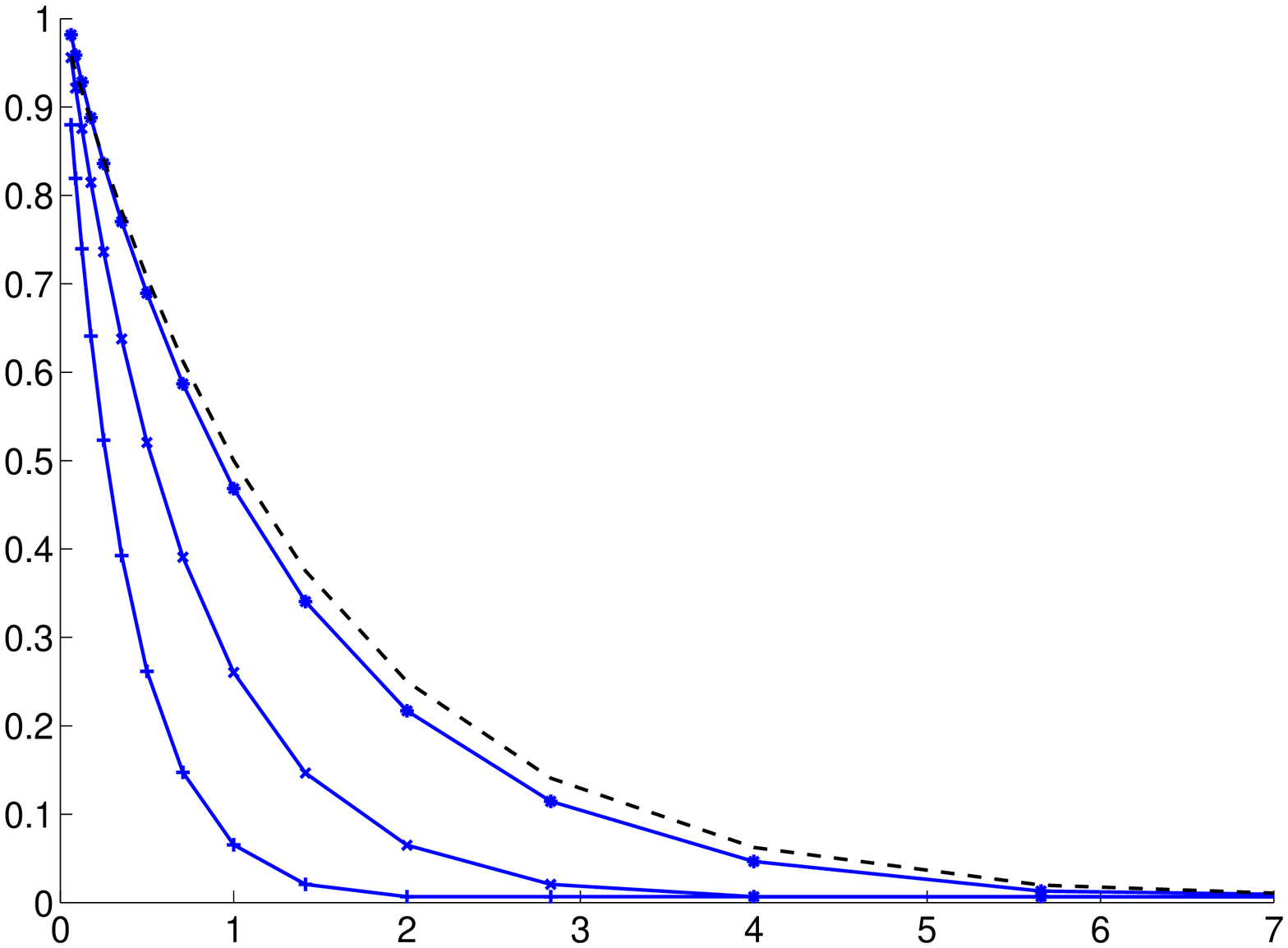}};
\node[rotate = -45] at (-1,-1.5) {\scriptsize $D_U(R)$};
\node[rotate = -45] at (-1.5,-1.8) {\scriptsize $0.9 W$};
\node[rotate = -50] at (-2.2,-2) { {\scriptsize $0.5 W$}};
\node[rotate = 310] at (-2.8,-2.4) {{\scriptsize $0.25W$}};
\draw[->] (-3.5,-3) -- (4,-3) node[right, fill = white] {$R$};
\draw[->] (-3.5,-3) -- node[above, rotate = 90, yshift = 0.4cm] {Distortion}(-3.5,3);
\draw (1,1) rectangle (3,3);
\draw[->] (2,1)  -- (2,2.5) node[above] {$p(t)$};
\draw[color = red] (1,1)--(1.2,1) -- (2,2.3) -- (2.8,1) -- (3,1); 

\draw (-1.5,1) rectangle (0.5,3);
\draw[->] (-0.5,1)  -- (-0.5,2.5) node[above] {$P(f)$};

\draw[red] plot[domain=-1:1, samples=50] (\x-0.5,{sin(deg(10*\x))*sin(deg(10*\x)) /(70*\x*\x)+1}) ;

\draw (1,-1.5) rectangle (3,0.5);
\draw[->] (2,-1.5) -- (2,0) node[above] {$S_U(f)$};
\draw[color = blue] (1,-1.5) node[below] {\scriptsize $-0.5W$} -- (1.2,-1.5) -- (1.2,-0.5) -- (2.8,-0.5) -- (2.8,-1.5) -- (3,-1.5) node[below] {\scriptsize $0.5W$}; 
\end{tikzpicture}
\caption{
The DRF of the PAM signal \eqref{eq:pam_def} for three values of sampling rate $1/T_0$ compared to the Nyquist rate $W$ of $U(\cdot)$. The DRF of the baseband stationary signal $U(\cdot)$ (assuming an energy preserving modulation) is given by the dashed curve. The PSD of $U(\cdot)$ and the shape of the pulse $p(t)$ are given in the small frames. This figure shows that the information content of in the PAM process decreases with the sampling rate. 
\label{fig:PAM_example}}
\end{center}
\end{figure}


\subsection{Amplitude Modulation with Random Phase} \label{subsec:AM}
In this section we turn back to the two processes discussed in the introduction as our motivating examples and evaluate their DRFs using Theorem~\ref{thm:main_result}. \\

Consider the process $X_\Phi(\cdot)$ obtained by modulating a stationary Gaussian process $U(\cdot)$ by a cosine-wave of frequency $f_0$ and a random phase $\Phi$ uniform over $[0,2\pi)$, as defined in  \eqref{eq:intro_random_phase}. It is an elementary exercise \cite[Ex. 8.18]{2009introduction} to show that the process $X_{\Phi}(\cdot)$ is stationary with PSD 
\begin{equation} 
S_{\Phi}(f) = \frac{1}{2} S_U(f-f_0)+ \frac{1}{2}S_U(f+f_0).
\end{equation}
From \cite[Thm. 4.6.5]{berger1971rate}, an upper bound on the DRF of $X_\Phi(t)$, denoted by $D_{X_\Phi}(R)$, is obtained by the DRF of a Gaussian process with the same PSD $S_{\Phi}(f)$ through the reverse-waterfilling \eqref{eq:SKP_cont}. However, it seems that $D_{X_\Phi}(R)$ cannot be determined solely from the second order statistics of $X_\Phi(\cdot)$. \par

The main obstacle in deriving $D_{X_\Phi}(R)$ is the random phase of $X_{\Phi}[\cdot]$, which makes the process non-Gaussian and non-ergodic. This random phase can be handled using an asynchronous block code \cite[Ch. 11.6]{gray1990entropy}, i.e. by adding a short prefix consisting of a source synchronization word to each block. Indeed, the following proposition follows directly from the proof of Theorem 11.6.1 in \cite{gray1990entropy}:
\begin{prop} \label{prop:am_to_rp}
For any $\varphi \in [0,2\pi)$ (deterministic), the DRF of the process $X_{\Phi}(\cdot)$ coincides with the DRF of the process 
\begin{equation} \label{eq:AM_def}
X_{\varphi}(t) = \sqrt{2} U(t) \cos\left(2\pi f_0 t+\varphi\right),\quad t\in \mathbb R.
\end{equation}
\end{prop}
It was noted in Example~\ref{ex:AM} above that $X_{\varphi}(\cdot)$ is CS with the SCD function \eqref{eq:AM_SCD}. It follows that $D_{X_\Phi}(R)$ is given by the DRF of the Gaussian CS process $X_\varphi(\cdot)$, generated by modulating the stationary Gaussian process $U(\cdot)$ using a deterministic cosine wave. Note that regardless of the carrier frequency $f_0$, the baseband process $U(\cdot)$ can always be recovered from $X_{\varphi}(\cdot)$, and that the $\sqrt{2}$ factor implies that the modulation preserves energy. These two facts are not enough to guarantee equality between the DRFs of the processes, since the modulation may lead to a `change in coordinates' in the spectrum, in analogy with \eqref{eq:change_of_coordinate} and \cite[Ch. 22]{Shannon1948}. In the following proposition we use Theorem~\ref{thm:main_result} to show that this equality indeed holds as long as $f_0$ is bigger than twice the bandwidth of $S_U(f)$. 
\begin{prop} \label{prop:DRF_AM}
Let $U(\cdot)$ be a Gaussian stationary process bandlimited to $\left(-f_B,f_B\right)$. Let $f_0>2f_B$. The DRF of the process 
\[
X_\varphi(t) = \sqrt{2} U(t) \cos\left(2\pi f_0 t+\varphi\right),\quad t\in \mathbb R,
\]
equals the DRF of the stationary Gaussian process $U(\cdot)$. 
\end{prop}
\begin{proof}
See Appendix~\ref{sec:proof_DRF_AM}. 
\end{proof}
Proposition~\ref{cor:AM_DRF} asserts that the process $X_{\varphi}(\cdot)$ with AM signal structure suffers the same minimal distortion as the baseband process $U(\cdot)$ upon the encoding of each of them at rate $R$, and provided the latter is narrowband. Figure~\ref{fig:example_AM} shows that the above equality does not necessarily hold when $U(\cdot)$ is not narrowband. 
Propositions \ref{prop:am_to_rp} and \ref{prop:DRF_AM} leads to the following conclusion:
\begin{cor} \label{cor:AM_DRF}
Let $U(\cdot)$ be a Gaussian stationary process bandlimited to $\left(-f_B,f_B\right)$. Assume that $\Phi$ is uniformly distributed over $(0,2\pi)$ and $f_0>2f_B$. The distortion rate function of the stationary process 
\[
X_{\Phi}(t) = \sqrt{2} U(t) \cos\left(2\pi f_0 t+\Phi\right),\quad t\in \mathbb R,
\]
equals the DRF of the baseband process $U(\cdot)$. 
\end{cor}

It is interesting to note that the DRF of a Gaussian process with the same PSD as the stationary process $X_{\Phi}(\cdot)$ is strictly bigger than the DRF of the baseband process $U(\cdot)$, and therefore provides an upper bound to $D_{X_{\Phi}}(R)$. This upper bound is illustrated in Fig.~\ref{fig:example_AM}. 


\begin{figure}
\begin{center}
\begin{tikzpicture}
\node at (0.215,-0.25) {\includegraphics[scale=0.4]{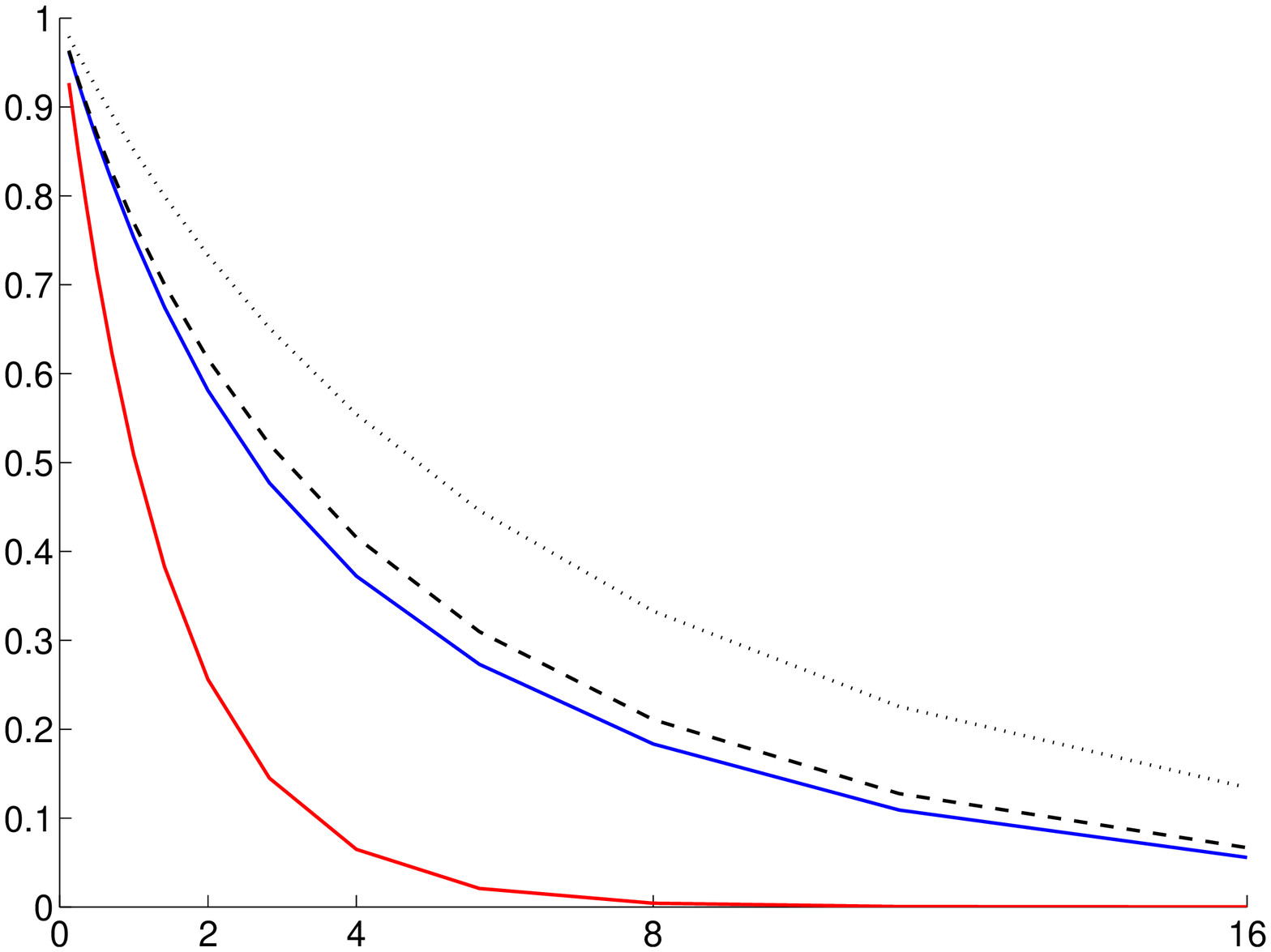}};
\node[rotate = -35] at (-1,-0.9) {\scriptsize $D_U(R)$};
\node[rotate = -35, color = blue] at (-1.4,-1.4) {\scriptsize $D_{X_\varphi}(R)$};
\draw[->] (-3.5,-3) -- (3.9,-3) node[right, fill = white] {$R$};
\draw[->] (-3.5,-3) -- node[above, rotate = 90, yshift = 0.4cm] {Distortion}(-3.5,3);
\draw (1,1) rectangle (3,3);
\draw[->] (2,1)  -- (2,2.5) node[above] {$S_U(f)$};
\draw[color = blue] (1,1)-- (1.5,1) node[below] {\scriptsize $-f_B$}  -- (2,2.3) -- (2.5,1) node[below] {\scriptsize $f_B$} -- (3,1); 

\draw (-1.5,1) rectangle (0.5,3);
\draw[->] (-0.5,1)  -- (-0.5,2.5) node[above] {$S_{\Phi}(f)$};

\draw[dotted] (-1.5,1) node[below] {\scriptsize $-f_0-f_B$} -- (-1.3,1)--(-0.7,1.7) -- (-0.1,1);
\draw[dotted]  (-0.9,1) --(-0.3,1.7) -- (0.3,1) node[below, xshift = 0cm] {\scriptsize $f_0+f_B$} ;
\draw (-1.5,1) -- (-1.3,1)--(-0.9,1.5)--(-0.7,2) -- (-0.3,2) -- (-0.1,1.5)--(0.3,1)--(0.5,1);

\end{tikzpicture}
\caption{\label{fig:example_AM} The DRF of the processes $X_{\varphi}(\cdot)$ of \eqref{eq:AM_def} (blue), the DRF of the baseband process $U(\cdot)$ (dashed), and the lower bound of Proposition~\ref{prop:lower_bound_cont}. The PSD $S_U(f)$ is taken to be the pulse given in the small frame. Proposition~\ref{prop:DRF_AM} implies that $D_U(R)$ and $D_{X_\varphi}(R)$ coincides for $f_s > 2f_B$. Also shown is the DRF of the Gaussian stationary process with PSD $S_\Phi(f)$ (dotted), which gives an upper bound to $D_{X_\varphi}(R)$. 
}
\end{center}
\end{figure}

\section{Conclusions} \label{sec:conclusions}
We derived an expression for the distortion-rate function (DRF) of a class of Gaussian processes with periodically time-varying statistics, known as cyclostationary (CS) processes. This DRF is computed by reverse waterfilling over eigenvalues of a spectral density matrix associated with the polyphase components in the decomposition of the source. Unlike other general expression for the DRF of Gaussian processes that use orthogonal basis expansion over increasing but finite time intervals, the expression we derive exploits the CS of the process by orthogonalizing the polyphase components. Since these components are defined over the entire time horizon, the resulting expression can be expressed in terms of the spectrum of the process. In the continuous-time counterpart the solution is given in terms of a limit over functions of these eigenvalues. \par
While we leave open the possibility whether there exists a closed form solution to the above limit in general, we have evaluated this limit in two special cases: a Gaussian CS processes with a PAM signal structure, and a Gaussian CS process with an amplitude modulation signal structure. As a result, we obtained the DRF of the processes obtained by these two important modulation schemes in terms of the power spectral density of the baseband stationary processes. We have also used the DRF result for a process with a PAM structure to derive the DRF of a process under combined sampling and source coding. 
\par
In addition to an expression for the DRF of CS processes, we have derived a lower bound on this DRF obtained by averaging the minimal distortion attained in encoding each of the polyphase components over a single period. This bound is tight when high correlation among these components is present. 

\appendices

\section{ \label{sec:main_result_discrete}}
In this Appendix we provide a proof of Theorem~\ref{thm:main_result_discrete}. 
Consider the vector valued process  ${\mathbf X}^M[\cdot]$ defined in \eqref{eq:associated_stationary}. The rate-distortion function of $\mathbf{X}^M[\cdot]$ is given by \eqref{eq:vector_RD}:
\begin{subequations}
\label{eq:main_result}
\begin{align}
D(\theta)& =\frac{1}{M}\sum_{m=1}^{M}\int_{-\infty}^\infty \min\left\{\lambda_m \left(e^{2\pi i\phi}\right) ,\theta \right\}d\phi, \\
R(\theta) & =\frac{1}{2} \sum_{m=1}^{M} \int_{-\infty}^\infty\log^+ \left[\lambda_m \left(e^{2\pi i\phi}\right)/\theta \right]d\phi,
\end{align} 
\end{subequations} 
where $0\leq \lambda_1\left(e^{2\pi i\phi}\right)\leq \ldots \leq \lambda_{M}\left(e^{2\pi i\phi}\right)$ are the eigenvalues of the spectral density matrix $\mathbf {S}_{\mathbf{{X}}^M}\left(e^{2\pi i \phi }\right)$ obtained by taking the Fourier transform of covariance matrix $\mathbf R_{\mathbf{ {X}}}[k]=\mathbb E\left[X^M[n+k] (X^M[n])^T \right]$ entry-wise. The $(m,r)^\textrm{th}$ entry of $\mathbf {S}_{\mathbf{{X}}^M}\left(e^{2\pi i \phi }\right)$ is given by \eqref{eq:PSD_discrete}:
\begin{align}
\left(\mathbf{S}_{\mathbf{{X}}^M}\left(e^{2\pi i \phi }\right) \right)_{m,r} & = S_{{X}}^{m,r}\left(e^{2\pi i \phi }\right) \nonumber \\
& = \frac{1}{M}\sum_{k=0}^{M-1} S_X^r \left(e^{2\pi i \frac{\phi-k}{M}} \right) e^{2\pi i (m-r) \frac{\phi-k}{M}}. \label{eq:main_discerete_proof_1} 
\end{align}

It is left to show that the DRF of $\mathbf X^M[\cdot]$ coincides with the DRF of $X[\cdot]$. By the source coding theorem for AMS processes \cite[Thm. 11.4.1]{gray1990entropy} it is enough to show that the operational block coding distortion-rate function (\cite[Ch. 11.2]{gray1990entropy}) of both processes is identical. Indeed, any $N$ block codebook for $\mathbf X^M[\cdot]$ is an $MN$ block codebook for $X[\cdot]$ which achieves the same quadratic distortion averaged over the block. However, since $\mathbf X^M[\cdot]$ is stationary, by \cite[Lemma. 11.2.3]{gray1990entropy} we know that any distortion above the DRF of $\mathbf X^M[\cdot]$ is attained for large enough $N$. This implies that the same is true for $X[\cdot]$. 

\section{ \label{sec:proof_main} }
In this Appendix we prove Theorem~\ref{thm:main_result}. 
Given a Gaussian cyclostationary process $X(\cdot)$ with period $T_0>0$, we define the discrete-time process $\bar{X}[\cdot]$ obtained by uniformly sampling $X(\cdot)$ at intervals $T_0/M$, i.e. 
\begin{equation} \label{eq:sampled_process_def}
\bar{X}[n]= X(nT_0/M),\quad n\in \mathbb Z.
\end{equation}
The autocorrelation function of $\bar{X}[\cdot]$ satisfies
\begin{align*}
 R_{\bar{X}}[n+M,k]& =\mathbb E\left[{\bar{X}}[n+M+k]{\bar{X}}[n+M] \right] \\
 & = \mathbb E\left[X(nT_0/M+T_0+k T_0/M)X(nT_0/M+T_0) \right] \\
 & = R_X(nT_0/M+T_0,kT_0/M+T_0) \\ 
 & = R_X(nT_0/M,kT_0/M) \\
 & = R_{\bar{X}}[n,k],
\end{align*}
which means that $\bar{X}[\cdot]$ is a discrete-time Gaussian cyclostationary process with period $M$.
The TPSD of $\bar{X}[\cdot]$ is given by
\[
S_{\bar{X}}^m(e^{2\pi i \phi})=\frac{M}{T_0}\sum_{k\in \mathbb Z} S_X^{mT_0/M}\left(\frac{M}{T_0}(\phi-k) \right).
\]
This means that the PSD of the $m^\textrm{th}$ PC of $\bar{X}[\cdot]$ is
\begin{align*}
S_{\bar{X}}^m\left(e^{2\pi i \phi} \right) & = \frac{1}{M} \sum_{n=0}^{M-1} S_{\bar{X}}^m\left(e^{2\pi i \frac{\phi-n}{M}} \right) \\
& = \frac{1}{T_0}  \sum_{n=0}^{M-1} \sum_{k\in \mathbb Z} S_X^{mT_0/M}\left(\frac{\phi-Mk-n}{T_0} \right) \\
& = \frac{1}{T_0} \sum_{l\in \mathbb Z} S_X^{mT_0/M} \left(\frac{\phi-l}{T_0} \right).
\end{align*}

By applying Theorem~\ref{thm:main_result_discrete} to $\bar{X}[\cdot]$, we obtain an expression for the DRF of $\bar{X}[\cdot]$ as a function of $M$:
\begin{subequations}
\label{eq:proof_main_cont}
\begin{align}
D_M(\theta_M)& =\frac{1}{M}\sum_{m=1}^{M}\int_{-\frac{1}{2}}^{\frac{1}{2}} \min\left\{\lambda_m\left(e^{2\pi i \phi }\right),\theta_M \right\}d\phi \label{eq:proof_main_cont_D} \\
\bar{R}(\theta_M)& =\frac{1}{2M} \sum_{m=1}^{M} \int_{-\frac{1}{2}}^{\frac{1}{2}}\log^+ \left[\lambda_m \left(e^{2\pi i \phi }\right)/ \theta_M \right]d\phi, \label{eq:proof_main_cont_R}
\end{align}
\end{subequations}
where $\lambda_1\left(e^{2\pi i \phi }\right) \leq \ldots \leq \lambda_{M}\left(e^{2\pi i \phi }\right)$ are the eigenvalues of the matrix with $(m+1,r+1)^\textrm{th}$ entry 
\begin{align}
S_{\bar{X}^m \bar{X}^r} \left(e^{2\pi i \phi}\right)& = \frac{1}{M}\sum_{n=0}^{M-1}S^r_{\bar{X}}\left(e^{2\pi i \frac{\phi-n}{M}} \right)e^{2\pi i (m-r) \frac{\phi-n}{M}} \label{eq:main_proof_final} \\
& = \frac{1}{T_0}\sum_{n=0}^{M-1} \sum_{k\in \mathbb Z } S_{X}^{rT_0/M} \left(\frac{\phi-n-kM}{T_0} \right) e^{2\pi i (m-r) \frac{\phi-n}{M}}, \nonumber \\
& = \frac{1}{T_0}\sum_{l\in \mathbb Z } S_{X}^{rT_0/M} \left(\frac{\phi-l}{T_0} \right) e^{2\pi i (m-r) \frac{\phi-l}{M}}. \nonumber
\end{align}
In order to express the code-rate in bits per time unit, we multiply the number of bits per sample $\bar{R}$ by the sampling rate $M/T_0$. This shows that the DRF of $\bar{X}[\cdot]$, as measured in bits pertime unit $R$, is given by \eqref{eq:main_result_cont}. \\

In order to complete the proof we rely on the following lemma:
\begin{lem} \label{lem:main_proof}
Let $X(\cdot)$ be as in Theorem~\ref{thm:main_result} and let $\bar{X}[\cdot]$ be its uniformly sampled version at rate $M/T_0$ as in \eqref{eq:sampled_process_def}. Denote the DRF at rate $R$ bits per time unit of the two processes by $D(R)$ and $\bar{D}(R)$, respectively. Then
\[
\lim_{M\rightarrow \infty} \bar{D}(R) = D(R). 
\]
\end{lem}
The rest of the appendix is devoted to the proof of Lemma~\ref{lem:main_proof}. \\

Throughout the next steps it is convenient to use the covariance kernels $K(t,s) = R_X(s,t-s)$ and $\bar{K}[n,k] = R_{\bar{X}}[n,k-n]$. For $M \in \mathbb N$, define 
\[
\widetilde{K}(t,s) = K\left(\lfloor tM/T_0 \rfloor T_0/M , \lfloor s M/T_0 \rfloor T_0/M \right). 
\]
For any fixed $T>0$, the kernel $\widetilde{K}(t,s)$ defines an Hermitian positive compact operator \cite{konig2013eigenvalue} on the space of square integrable functions over $[-T,T]$. The eigenvalues of this operator are given by the Fredholm integral equation
\begin{equation} \label{eq:main_proof_fredholm}
\tilde{\lambda}_l \tilde{f}_l(t) = \frac{1}{2T} \int_{-T}^T \widetilde{K}(t,s) \tilde{f}_l(s) ds,\quad -T \leq t \leq T,
\end{equation}
where it can be shown that there are at most $MT/T_0$ non-zero eigenvalues $\{ \tilde{\lambda}_l \}$ that satisfy \eqref{eq:main_proof_fredholm}. We define the function $\tilde{D}_T(R)$ by the following parametric expression:
\begin{align}
\begin{split}
\tilde{D}_T(\theta) & = \sum_{l=1}^\infty \min \left\{\tilde{\lambda}_l,\theta \right\} \\
R(\theta) & = \frac{1}{2} \sum_{l=1}^\infty \log^+ \left(  \frac{\tilde{\lambda}_l}{\theta} \right)
\label{eq:main_proof_DRF_approx}
\end{split} 
\end{align}
(the eigenvalues in \eqref{eq:main_proof_DRF_approx} are implicitly depend on $T$). 
Note that 
\begin{equation}
\sum_{l=1}^\infty \tilde{\lambda}_l = \frac{1}{2T} \int_{-T}^T \widetilde{K}(t,t) dt = \frac{1}{2T} \sum_{n=-N}^N K(nT_0/M,nT_0/M), \label{eq:proof_eigs_sum}
\end{equation}
where $N = MT/T_0$. Expression \eqref{eq:proof_eigs_sum} converges to 
\[
\frac{1}{2T}\int_{-T}^T K(t,t) dt \leq \sigma_X^2
\]
as $M$ goes to infinity due to our assumption that $R(t,\tau)$ is Riemann integrable and therefore so is $K(t,s)$. Since we are interested in the asymptotic of large $M$, we can assume that \eqref{eq:proof_eigs_sum} is bounded. This implies that $\tilde{D}_T(R)$ is bounded. 

We would like to claim that the eigenvalues $ \{ \tilde{\lambda}_l \}$ approximate the eigenvalues $\{ \lambda_l \}$. We have the following lemma:
\begin{lem} \label{lem:main_proof_1}
Let $\{ \lambda_l \}$ and $\{ \tilde{\lambda}_l \}$ be the eigenvalues in the Fredholm integral equation of $K(t,s)$ and $\widetilde{K}(t,s)$, respectively. Assume that these eigenvalues are numbered in a descending order. Then
\begin{align}
 \label{eq:main_proof_Lipschitz}
\left| \lambda_l - \tilde{\lambda}_l \right|  \leq 4 C T_0/M, \quad l=1,2,\ldots.
\end{align}
\end{lem}
\subsubsection*{Proof of Lemma~\ref{lem:main_proof_1}}
Approximations of the kind \eqref{eq:main_proof_Lipschitz} can be obtained by Weyl's inequalities for singular values of operators defined by self-adjoint kernels \cite{weyl1912asymptotische}. In our case it suffices to use the following result \cite[Cor. 1'']{wielandt1956error}: 
\begin{equation} \label{eq:main_proof_Weyl}
\left| \lambda_l - \tilde{\lambda}_l \right| \leq 2 \sup_{t,s \in [-T,T]} \left| K(t,s) - \widetilde{K}(t,s) \right|, \quad l=1,2,\ldots.
\end{equation}
The assumption that $R_X(t,\tau)$ is Lipschitz continuous in $\tau$ implies that there exists a constant $C>0$ such that for any $t_1,t_2, s \in \mathbb R$, 
\[
\left| K(t_1,s)-K(t_2,s) \right| = \left| R_X(s,t_1-s)-R_X(s,t_2-s) \right|  \leq C \left| t_1-t_2 \right|. 
\]
We therefore conclude that $K_X(t,s)$ is Lipschitz continuous in both of its arguments from symmetry. Lipschitz continuity of $K(t,s)$ implies 
\begin{align*}
& \left| K(t_1,s_1) - K(t_2,s_2) \right| \\
& \leq \left| K(t_1,s_1) - K(t_1,s_2) \right| + \left| K(t_1,s_2) - \widetilde{K}(t_2,s_2) \right|  \\
 & \leq C \left| s_1 - s_2 \right| + C \left| t_1 - t_2 \right|. 
\end{align*}
As a result, \eqref{eq:main_proof_Weyl} leads to
\begin{align*}
\left| \lambda_l - \tilde{\lambda}_l \right| & \leq 2 \sup_{t,s} \left| K(t,s) - \widetilde{K}(t,s) \right| \\
& = 2 \sup_{t,s \in [-T,T]} \left| K(t,s) - K\left(\lfloor tM/T_0 \rfloor T_0/M , \lfloor s M/T_0 \rfloor T_0/M \right) \right| \\
& \leq 2C \left( \left| t -  \lfloor tM/T_0 \rfloor T_0/M \right| + \left| t -  \lfloor s M/T_0 \rfloor T_0/M \right| \right) \\
& \leq 4 C T_0/M,
\end{align*} 
which proves Lemma~\ref{lem:main_proof_1}. \\

The significance of Lemma~\ref{lem:main_proof_1} is that the eigenvalues of the kernel $K(t,s)$ used in the expression for the DRF of $X(\cdot)$ can be approximated by the eigenvalues of $\widetilde{K}(t,s)$, where the error in each of these approximations converge, uniformly in $T$, to zero as $M$ increases. Since only a finite number of eigenvalues participate in \eqref{eq:DRF_KL} and since  both $D_T(R)$ and $\tilde{D}_T(R)$ are bounded continuous functions of their eigenvalues, we conclude that $\tilde{D}_T(R)$ converges to $D_T(R)$ uniformly in $T$. \\
Now let $\epsilon > 0$ and fix $M_0$ large enough such that for all $M>M_0$ and for all $T$ 
\begin{equation}
\label{eq:main_proof_part1}
\left| D_T(R) - \tilde{D}_T(R) \right| \leq \epsilon. 
\end{equation}

Recall that in addition to \eqref{eq:SKP}, the DRF of $\bar{X}[\cdot]$, denoted here as $\bar{D}(\bar{R})$, can also be obtained as the limit in $N$ of the expression
\begin{align*}
\bar{D}_N(\theta) & = \sum_{l=1}^\infty \min \left\{ \bar{\lambda}_l,\theta \right\} \\
\bar{R}(\theta) & = \frac{1}{2} \sum_{l=1}^\infty \log^+ \left(\bar{\lambda}_l/\theta \right),
\end{align*}
where $\bar{\lambda}_1,\bar{\lambda}_2,\ldots$ are the eigenvalues in the KL expansion of $\bar{X}$ over $n=-N,\ldots,N$:
\begin{align} \label{eq:main_proof_discrte_KL}
\bar{\lambda}_l f_l[n] = \frac{1}{2N+1} \sum_{k=-N}^{N} K_{\bar{X}}[n,k] f_l[k],\quad l=1,\ldots, N,
\end{align}
(there are actually at most $2N+1$ distinct non-zero eigenvalues that satisfies \eqref{eq:main_proof_discrte_KL}). Letting $T_N = T_0 M / N$ and $\tilde{f}_l(t) = f_l\left( \lfloor t/T_0 \rfloor M \right)$ \eqref{eq:main_proof_discrte_KL} can also be written as
\begin{align*}
\bar{\lambda}_l f_l[n] & = \int_{-T_N}^{T_N} \widetilde{K}_{X}(nT_0/M,s)  f_l\left[ \lfloor s/T_0  \rfloor M \right] ds,\quad l=1,2,\ldots, \\
\bar{\lambda}_l \tilde{f}_l(t) & = \int_{-T_n}^{T_N} \widetilde{K}(t,s)  \tilde{f}_l(s) ds,\quad -T_N< t < T_N. 
\end{align*}

From the uniqueness of the KL expansion, we obtain that for any $N$, the eigenvalues of $\widetilde{K}(t,s)$ over $T_N = T_0M/N$ are given by the eigenvalues of $\bar{K}[n,k]$ over $-N,\ldots,N$. We conclude that 
\begin{equation}
\label{eq:main_proof_part2} 
\bar{D}_N(\bar{R}) = \tilde{D}_{T_N}(R),
\end{equation} 
where $R = \bar{R} T_0 / M$. Now take $N$ large enough such that 
\[
\left| \bar{D}_N(R) - \bar{D}(R) \right| < \epsilon,
\]
and 
\[
\left| D_{T_N}(R) - D(R) \right| < \epsilon.
\]
For all $M\geq M_0$ we have
\begin{align}
 \left| D(R) - \bar{D}(R) \right| & =  \left| D(R) - D_{T_N}(R) + D_{T_N}(R) + \tilde{D}_{T_N}(R) \right. \nonumber \\
 & ~~ - \tilde{D}_{T_N}(R)   \left.+ \bar{D}_N(R) - \bar{D}_N(R) - \bar{D}(R) \right| \nonumber \\ 
& \leq  \left| D(R) - D_{T_N}(R) \right|  \label{eq:main_proof_l1} \\
& + \left| D_{T_N}(R) -  \tilde{D}_{T_N}(R)\right|  \label{eq:main_proof_l2} \\
& + \left| \tilde{D}_{T_N}(R)  - \bar{D}_N(R) \right|  \label{eq:main_proof_l3} \\
& + \left| \bar{D}_N(R) - \bar{D}(R) \right|  \leq 3\epsilon,  \label{eq:main_proof_l4} 
\end{align}
where the last transition is because: \eqref{eq:main_proof_l1} and \eqref{eq:main_proof_l4} are smaller than $\epsilon$ by the choice of $N$, 
\eqref{eq:main_proof_l2} is smaller than $\epsilon$ from \eqref{eq:main_proof_part1}. 
and \eqref{eq:main_proof_l3} equals zero from \eqref{eq:main_proof_part2}.

\section{ \label{sec:proof_bound_cont} }
In this Appendix we provide a proof of Proposition~\ref{prop:lower_bound_cont}. 
We use the process $\bar{X}[\cdot]$ defined in the proof of Theorem~\ref{thm:main_result} as the uniform sampled version of $X(\cdot)$ at rate $T_0/M$. From Proposition~\ref{prop:lower_bound} we conclude that the DRF of $\bar{X}[\cdot]$ satisfies
\begin{equation}
\label{eq:lower_bound_cont_proof}
D_{\bar{X}}(\bar{R}) \geq \frac{1}{M} \sum_{m=0}^{M-1} \int_{-\frac{1}{2}}^\frac{1}{2} \min \left\{  \frac{1}{T_0} \sum_{l\in \mathbb Z} S_X^{mT_0/M} \left(\frac{\phi-l}{T_0} \right),\theta_m \right\} d\phi,
\end{equation}
where for all $m=0,\ldots,M-1$, $\theta_m$ is determined by
\[
\bar{R} = \frac{1}{2} \int_{-\frac{1}{2}}^\frac{1}{2} \log^+ \left[ \frac{1}{T_0} \sum_{l\in \mathbb Z} S_X^{mT_0/M} \left(\frac{\phi-l}{T_0} \right)/\theta_m\right] d\phi.
\]
Denote $t=mT_0/M$. As $M$ approaches infinity, the RHS of \eqref{eq:lower_bound_cont_proof} converges to an integral with respect to $t$ over the interval $(0,T_0)$, which implies 
\begin{align} \label{eq:bound_proof_D}
\bar{D}(\bar{R}) &  \geq   \frac{1}{T_0} \int_0^{T_0} \int_{-\frac{1}{2}}^\frac{1}{2} \min \left\{  \sum_{l\in \mathbb Z} S_X^{~t} \left(\frac{\phi-l}{T_0} \right),\theta_t \right\} d\phi,
\end{align}
and 
\begin{align} \label{eq:bound_proof_R}
\bar{R} = \frac{1}{2} \int_{-\frac{1}{2}}^\frac{1}{2} \log^+ \left[ \sum_{l\in \mathbb Z} S_X^{mT_0/M} \left(\frac{\phi-l}{T_0} \right)/\theta_m\right] d\phi,
\end{align}
where we denoted $\theta_t = T_0 \theta_m$. In order to go from $\bar{R}$ to $R$ we multiply \eqref{eq:bound_proof_R} by $M/T_0$, so that \eqref{eq:bound_proof_D} and \eqref{eq:bound_proof_R} lead to \eqref{eq:lower_bound_cont}. The fact that the function $\bar{D}(R)$ converges to $D(R)$ as $M$ goes to infinity follows from the proof of Theorem.~\ref{thm:main_result}. 

\section{\label{app:pam_proof}}
In this Appendix we provide a proof of Proposition~\ref{prop:pam}. The entries of the matrix $\mathbf S \expphi$ in Theorem~\ref{thm:main_result} are obtained by using the CPSD of the PAM process \eqref{eq:PAM_CPSD} in \eqref{eq:main_result_entries}. For all $M \in \mathbb N$, this leads to
\begin{align}
& \mathbf S_{m+1,r+1} \expphi  = \nonumber \\ 
& \frac{1}{T_0^2} \sum_{k\in \mathbb Z} \left[ P\left(\frac{\phi-k}{T_0} \right) S_U \left( \frac{\phi-k}{T_0} \right) e^{2\pi i (\phi-k)\frac{m-r}{M} } \right. \nonumber \\
& \quad \quad \times \left. \sum_{n\in \mathbb Z} P^* \left( \frac{\phi-n-k}{T_0}\right) e^{2\pi i \frac{nr}{M}} \right] \label{eq:pam_proof} \\
&= \frac{1}{T_0^2} \sum_{k\in \mathbb Z}  P\left(\frac{\phi-k}{T_0} \right) S_U \left( \frac{\phi-k}{T_0} \right) e^{2\pi i (\phi-k)\frac{m}{M} }  \label{eq:pam_proof_a} \\
 & \quad  \quad \times \sum_{l\in \mathbb Z} P^* \left( \frac{\phi-l}{T_0}\right) e^{-2\pi i(\phi-l) \frac{r}{M}}. \nonumber 
\end{align}
The expression \eqref{eq:pam_proof_a} consist of the product of a term depending only on $m$ and a term depending only on $r$. We conclude that the matrix $\mathbf S \expphi$ can be written as the outer product of two $M$ dimensional vector, and thus it is of rank one. The single non-zero eigenvalue $\lambda_M \expphi$ of $\mathbf S \expphi$ is given by the trace of the matrix, which, by the orthogonality of the functions $e^{2\pi i \frac{nr}{M}}$ in \eqref{eq:pam_proof}, is evaluated as
\begin{align}
& \lambda_M \expphi = \frac{M}{T_0^2} \sum_{k\in \mathbb Z} \left |P\left(\frac{\phi-k}{T_0} \right) \right|^2 S_U \left( \frac{\phi-k}{T_0} \right). \label{eq:pam_proof_2}
\end{align}
We now use \eqref{eq:pam_proof_2} in \eqref{eq:main_result_cont}. In order to obtain \eqref{eq:DRF_PAM}, we change the integration variable from $\phi$ to $f=\phi/T_0$ and the water-level parameter $\theta$ to $T_0\theta/M$. Note that the final expression is independent of $M$, so the limit in \eqref{eq:main_thm_DR} is already given by this expression.

\section{ \label{sec:proof_DRF_AM}}
In this Appendix we provide a proof of Proposition~\ref{prop:DRF_AM}. Since $S_U(f)$ is compactly supported, the covariance function $R_U(\tau) = \mathbb E U(t+\tau) U(t)$ is an analytic function and therefore Lipshitz continuous. Lipschitz continuity of $R_U(\tau)$ implies Lipschitz continuity of $R_X(t,\tau)$ in its second argument and therefore Theorem~\ref{thm:main_result} applies: 
The DRF of the Gaussian CS process $X_\varphi(\cdot)$  with period $T_0=f_0^{-1}$ is obtained by using Theorem \ref{thm:main_result} with the SCD \eqref{eq:AM_SCD}. For all $M\in \mathbb N$ and $m,r =0,\ldots,M-1$ we have,
\begin{align*}
& \mathbf{S}_{m+1,r+1} \left(e^{2\pi i \phi} \right) = \\
&\frac{f_0}{2} \sum_{k\in \mathbb Z} S_U \left(f_0(\phi-k-1) \right) \left(1+e^{4\pi i r/M} \right)e^{2\pi i (\phi-k) \frac{m-r}{M}} \\
& ~~ + \frac{f_0}{2} \sum_{k\in \mathbb Z} S_U \left(f_0(\phi-k+1) \right) \left(1+e^{-4\pi i r/M} \right)e^{2\pi i (\phi-k) \frac{m-r}{M} }.
\end{align*}
Under the assumption that $f_0>2f_B$ we have that for all $\phi\in \left(-\frac{1}{2},\frac{1}{2}\right)$,  $S_U \left(f_0(\phi-k \pm 1) \right)=0$ for all $k\neq \pm1$. This leads to 
\begin{align}
 &\mathbf{S}_{m+1,r+1} \left(e^{2\pi i \phi} \right) \\ & =  S_U(f_0\phi)\frac{f_0\left(1+e^{4\pi i r/M} \right)}{2} e^{2\pi i(\phi+1)\frac{m-r}{M} } \nonumber \\
 &~~~~~~+ S_U(f_0\phi)\frac{f_0\left(1+e^{-4\pi i r/M} \right)}{2} e^{2\pi i(\phi-1)\frac{m-r}{M} } \nonumber \\
 & = 2f_0 S_U(f_0\phi) e^{2\pi i \frac{m-r}{M} \phi } \cos \left(2\pi \frac{m}{M} \right) \cos \left( 2\pi \frac{r}{M} \right).
 \label{eq:proof_am}
\end{align}

From \eqref{eq:proof_am} we conclude that the matrix $\mathbf S \left(e^{2\pi i \phi} \right)$ can be written as 
\[
\mathbf S \left(e^{2\pi i \phi} \right) = 2f_0 S_U(f_0 \phi) \mathbf S_M \left(e^{2\pi i \phi} \right) \mathbf S_M^* \left(e^{2\pi i \phi} \right),
\]
where $\mathbf S_M \left(e^{2\pi i \phi} \right) \in \mathbf R^{M \times 1}$ is given by
\[
\left(1, e^{2\pi i \phi/M}\cos\left(\frac{2\pi}{M} \right),\ldots,e^{2\pi i \phi \frac{M-1}{M}} \cos \left(\frac{2\pi(M-1)}{M} \right) \right).
\]
This means that $\mathbf S \left(e^{2\pi i \phi} \right)$ is a matrix of rank one, and its single non-zero eigenvalue is given by its trace:
\[
\lambda_{M} \left(e^{2\pi i \phi} \right) = 2f_0 S_U(f_0 \phi) \sum_{m=0}^{M-1} \cos^2\left(2\pi m/M\right) =M f_0 S_U(f_0 \phi).
\]
We use this in \eqref{eq:main_result_entries}:
\begin{align}
R_M(\theta) & = \frac{f_0}{2} \sum_{m=1}^{M} \int_{-\frac{1}{2}}^\frac{1}{2} \log^+ \left[\lambda_m\left(e^{2\pi \phi}\right) \right] d\phi \nonumber \\
& =  \frac{f_0}{2} \int_{-\frac{1}{2}}^\frac{1}{2} \log^+ \left[Mf_0 S_U(f_0 \phi)/\theta  \right] d\phi \nonumber \\
& =  \frac{1}{2}  \int_{-\infty}^\infty  \log^+ \left[ S_U(f)/(\theta/M)  \right] df \label{eq:proof_am_R},
\end{align}
and
\begin{align}
D_M(\theta) & = \frac{f_0}{M} \int_{-\frac{1}{2}}^\frac{1}{2} \min\left\{Mf_0 S_U(f_0 \phi),\theta \right\}d\phi \nonumber \\
& =f_0 \int_{-\infty}^{\infty} \min\left\{ S_U(f_0 \phi),\theta /M \right\}d\phi \nonumber \\
& = \int_{-\infty}^{\infty} \min\left\{ S_U(f),\theta /M \right\}df. \label{eq:proof_am_D}
\end{align}
From \eqref{eq:proof_am_R} and \eqref{eq:proof_am_D} we conclude that for every $M$, the parametric expression of $D$ as a function of $R$ is identical to the DRF of the stationary process $U(\cdot)$ given by \eqref{eq:SKP_cont}. 

\section*{ACKNOWLEDGMENT}
This work was supported in part by the NSF under grant CCF-1320628, under the NSF Center for Science of Information (CSoI) grant CCF-0939370, and under BSF Transformative Science Grant 2010505. The authors wish to thank Robert M. Gray for helpful remarks. We also wish to thank to anonymous reviewers and the AE for the depth and details of their reviews which have greatly improved this paper. 

\bibliographystyle{IEEEtran}
\bibliography{IEEEfull,/Users/Alon1/LaTex/bibtex/sampling}


\end{document}